\documentclass[onecolumn]{bmcart}

\usepackage[latin9]{inputenc}
\usepackage{float}

\usepackage{amsmath}

\usepackage{amssymb}
\usepackage{graphicx}
\usepackage{algorithm2e}

\makeatletter
\newcommand{\noop}[1]{}

\providecommand{\tabularnewline}{\\}
\floatstyle{ruled}
\newfloat{algorithm}{tbp}{loa}
\providecommand{\algorithmname}{Algorithm}
\floatname{algorithm}{\protect\algorithmname}

\@ifundefined{showcaptionsetup}{}{%
 \PassOptionsToPackage{caption=false}{subfig}}
 \usepackage{subfig}

\startlocaldefs
\newtheorem{theorem}{Theorem}
\newtheorem{lemma}[theorem]{Lemma}
\newtheorem{proposition}[theorem]{Proposition}

\newenvironment{proof}[1][Proof]{\begin{trivlist}
\item[\hskip \labelsep {\bfseries #1}]}{\end{trivlist}}

\endlocaldefs

\begin{document}

\begin{frontmatter}

\begin{fmbox}
\dochead{Research}


\title{Near-optimal Assembly for Shotgun Sequencing with Noisy Reads}


\author[
   addressref={aff1},                   
   email={kklam@eecs.berkeley.edu}   
]{\inits{KK}\fnm{Ka-Kit} \snm{Lam}}
\author[
   addressref={aff2},
   email={akhalak@pacificbiosciences.com} 
]{\inits{A}\fnm{Asif} \snm{Khalak}}
\author[
   addressref={aff1},
   corref={aff1,aff2},                    
   email={dnctse@gmail.com}
]{\inits{NC}\fnm{David} \snm{Tse}}


\address[id=aff1]{
  \orgname{Department of Electrical Engineering and Computer Sciences, UC Berkeley},
  \city{Berkeley, California},                              
  \cny{United States}                                    
}
\address[id=aff2]{
  \orgname{Pacific Biosciences},
  \city{Menlo Park, California},                              
  \cny{United States}   
}




\begin{abstractbox}

\begin{abstract} 
Recent work identified the fundamental limits on the information requirements in terms of read length and coverage depth required for successful {\em de novo} genome reconstruction from shotgun sequencing data, based on the idealistic assumption of no errors in the reads (noiseless reads). In this work, we show that even when there is noise in the reads, one can successfully reconstruct with information requirements close to the noiseless fundamental limit. A new assembly algorithm, X-phased Multibridging,  is designed based on a probabilistic model of the genome. It is shown through analysis  to perform well on the model, and through simulations to perform well  on real genomes.

\end{abstract}




\end{abstractbox}
\end{fmbox}

\end{frontmatter}

\section*{Background}

%
%
%
%
%
%
%

%
%
Optimality in the acquisition and processing of DNA sequence data represents a serious technology challenge from various perspectives including sample preparation, instrumentation and algorithm development.  Despite scientific achievements such as the sequencing of the human genome and ambitious plans for the future \cite{turnbaugh2007human,sequencing2007plan},  there is no single, overarching framework to identify the fundamental limits in terms of information requirements required for successful output of the genome from the sequence data.

Information theory has been successful in providing the foundation for such a framework in digital communication \cite{Sha48}, and we believe that it can also provide  insights into understanding the essential aspects of DNA sequencing. A first step in this direction has been taken in the recent work \cite{bresler2013optimal}, where the fundamental limits on the minimum read length and coverage depth required for successful assembly are identified in terms of the statistics of various repeat patterns in the genome. Successful assembly is defined as the reconstruction of the underlying genome, i.e. genome finishing~\cite{pop2009genome}. The genome finishing problem is particularly attractive for analysis because it is clearly and unambiguously defined and is arguably the ultimate goal in assembly. There is also a scientific need for finished genomes~\cite{medini2008microbiology}\cite{mardis2002finished}. Until recently, automated genome finishing was beyond reach \cite{gordon1998consed} in all but the simplest of genomes.  New advances using ultra-long read single-molecule sequencing, however, have reported successful automated finishing~\cite{chin2013nonhybrid,koren2012hybrid}. 
Even in the case where finished assembly is not possible, the results in \cite{bresler2013optimal} provide insights on optimal use of read information since 
the heart of the problem lies in  
how one can optimally use the read information to resolve repeats. 

Figure \ref{noiselessPlot} gives an example result for the repeat statistics of \textit{E. coli} K12. The x-axis of the plot is the read length and the y-axis is the coverage depth normalized by the Lander-Waterman depth (number of reads needed to cover the genome \cite{lander1988genomic}). The lower bound  identifies the necessary read length and coverage depth required for {\it any} assembly algorithm to be successful with these repeat statistics. An assembly algorithm called Multibridging Algorithm was presented, whose read length and coverage depth requirements are very close to the lower bound, thus tightly characterizing the fundamental information requirements. The result shows a critical phenomenon at a certain read length $L = \ell_{\rm crit}$: below this critical read length, reconstruction is impossible no matter how high the coverage depth; slightly above this read length, reconstruction is possible with Lander-Waterman coverage depth. This critical read length is given by $\ell_{\rm crit} = \max \{\ell_{\rm int}, \ell_{\rm tri}\}$, where $\ell_{\rm int}$ is the length of the longest pair of exact interleaved repeats and $\ell_{\rm tri}$  is the length of the longest exact triple repeat in the genome, and has its roots in earlier work by Ukkonen on Sequencing-by-Hybridization \cite{ukkonen1992approximate}.  The framework also allows the analysis of specific algorithms and the comparison with the fundamental limit; the plot shows for example the performance of the Greedy Algorithm and we see that its information requirement is far from the fundamental limit.

\begin{figure} 

\subfloat[Information requirement for noiseless reads]{\includegraphics[width=8cm]{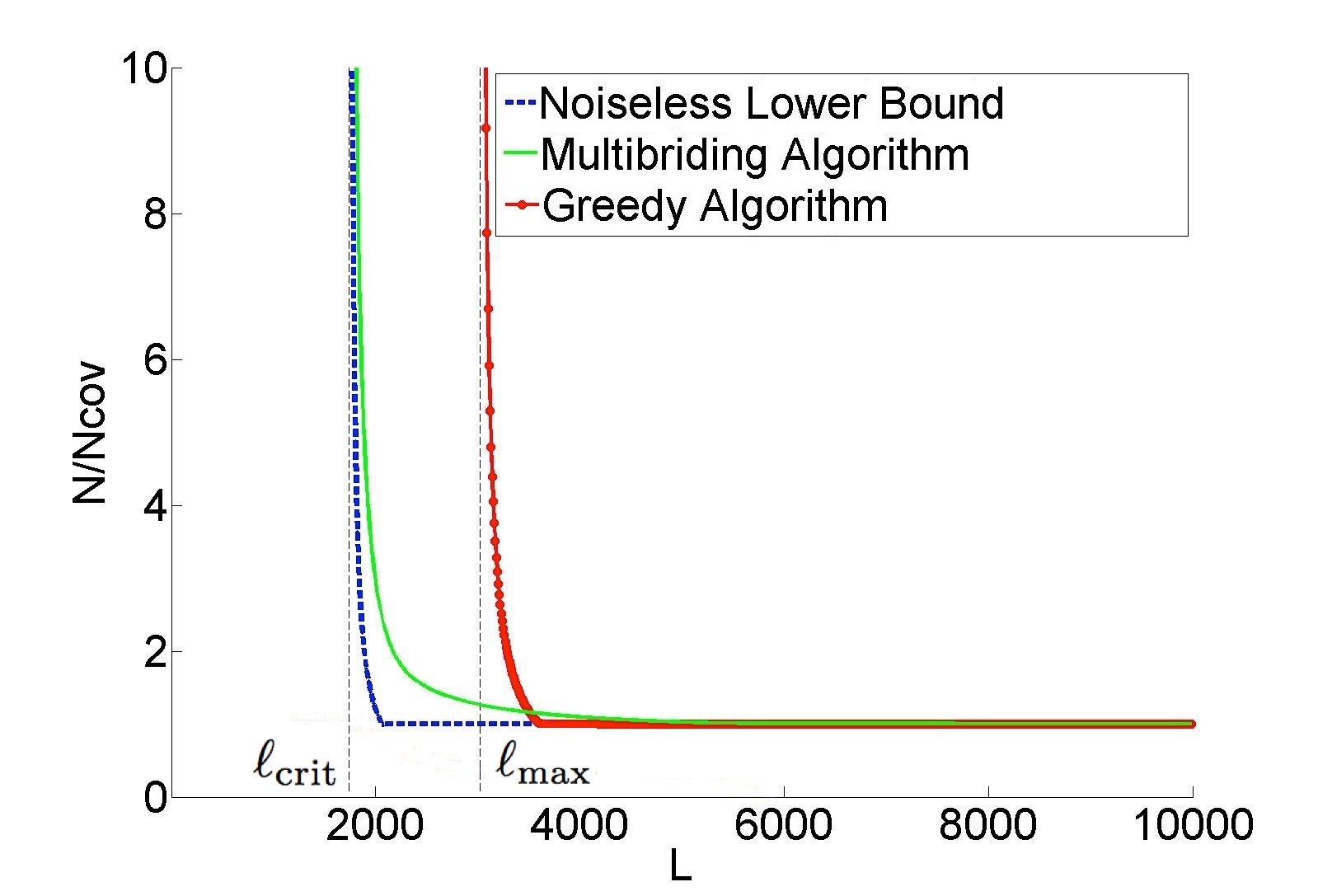}\label{noiselessPlot}}

\subfloat[Information requirement for noisy reads]{\includegraphics[width=8cm]{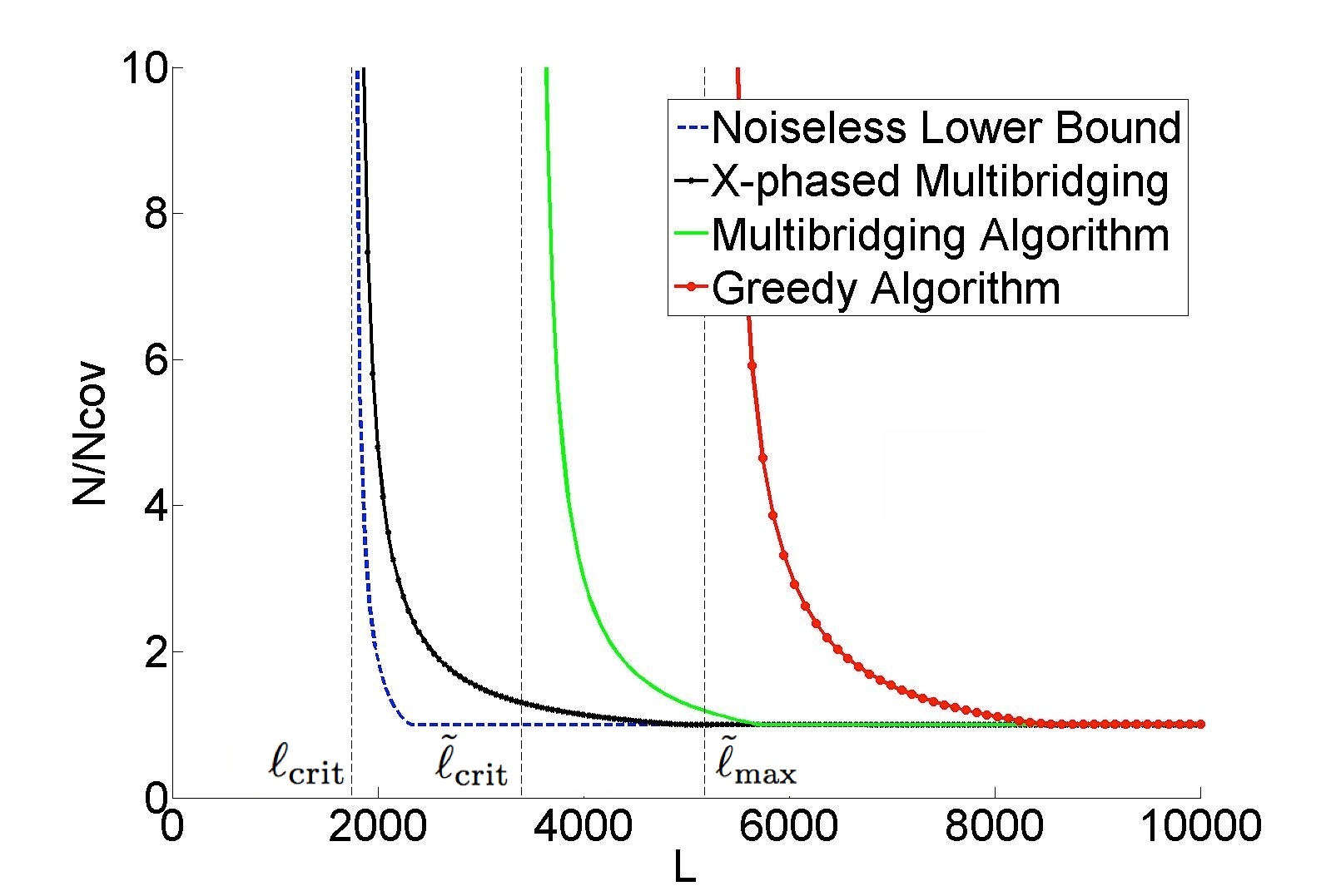}\label{noisyPlot}
}
\par
\begin{minipage}[t]{7cm}
\caption{Information requirement to reconstruct \textit{E. coli} K12. $\ell_{\rm crit} =1744, \tilde{\ell}_{\rm crit} = 3393$}
\end{minipage}
\end{figure}

A key simplifying assumption in \cite{bresler2013optimal} is that there are no errors in the reads (noiseless reads). However reads are noisy in all present-day sequencing technologies, 
%
ranging from primarily substitution errors in Illumina\textsuperscript{\textregistered} platforms, to primarily insertion-deletion errors in Ion Torrent\textsuperscript{\textregistered} and PacBio\textsuperscript{\textregistered} platforms. 
The following question is the focus of the current paper: in the presence of read noise, can we still successfully assemble with a read length and coverage depth close to the minimum in the noiseless case? A recent work \cite{poster}  with an existing assembler
suggests that the information requirement for genome finishing substantially exceeds the noiseless limit. However, it is not obvious whether the limitations lie in the fundamental effect of read noise or in the sub-optimality of the algorithms in the assembly pipeline.
%
%

\section*{Results}
The difficulty of the assembly problem depends crucially on the genome repeat statistics. Our approach to answering the question of the fundamental effect of read noise is based on design and analysis using a parametric probabilistic model of the genome that matches the key features of the repeat statistics we observe  in genomes. In particular, it models the presence of long flanked repeats which are repeats flanked by statistically uncorrelated region.
Figure~\ref{noisyPlot} shows a plot of the predicted information requirement 
for reliable reconstruction by various algorithms under a substitution error rate of $1\%$. The plot is based on analytical formulas  derived under our genome model with parameters set to match the statistics of \textit{E. coli} K12.  We  show that it is possible in many cases to develop algorithms that approach the noiseless lower bound even when the reads are noisy. 
Specifically, the  X-phased Multibridging Algorithm  has close to the same critical read length $L = \ell_{\rm crit}$ as in the noiseless case and only slightly greater coverage depth requirement for read lengths greater than the critical read length. 

We then proceed to build a prototype assembler based on the analytical
insights and we perform experiments on real genomes. As
shown in Figure~\ref{fig:Sim}, we test the prototype assembler by using it to assemble
noisy reads sampled from 4 different genomes. 
At coverage and
read length indicated by a green circle, we successfully assemble noisy reads into one contig (in most cases with
more than 99\% of the content matched when compared with the ground truth).
Note that the information requirement is close to the noiseless
lower bound. Moreover, the algorithm (X-phased Multibridging) is computationally efficient
with the most computational expensive step being the computation
of overlap of reads/K-mers, which is an unavoidable procedure
in most assembly algorithms.

\begin{figure} [!htb]
\subfloat[\textit{Prochlorococcus marinus} ]{\includegraphics[width=8cm]{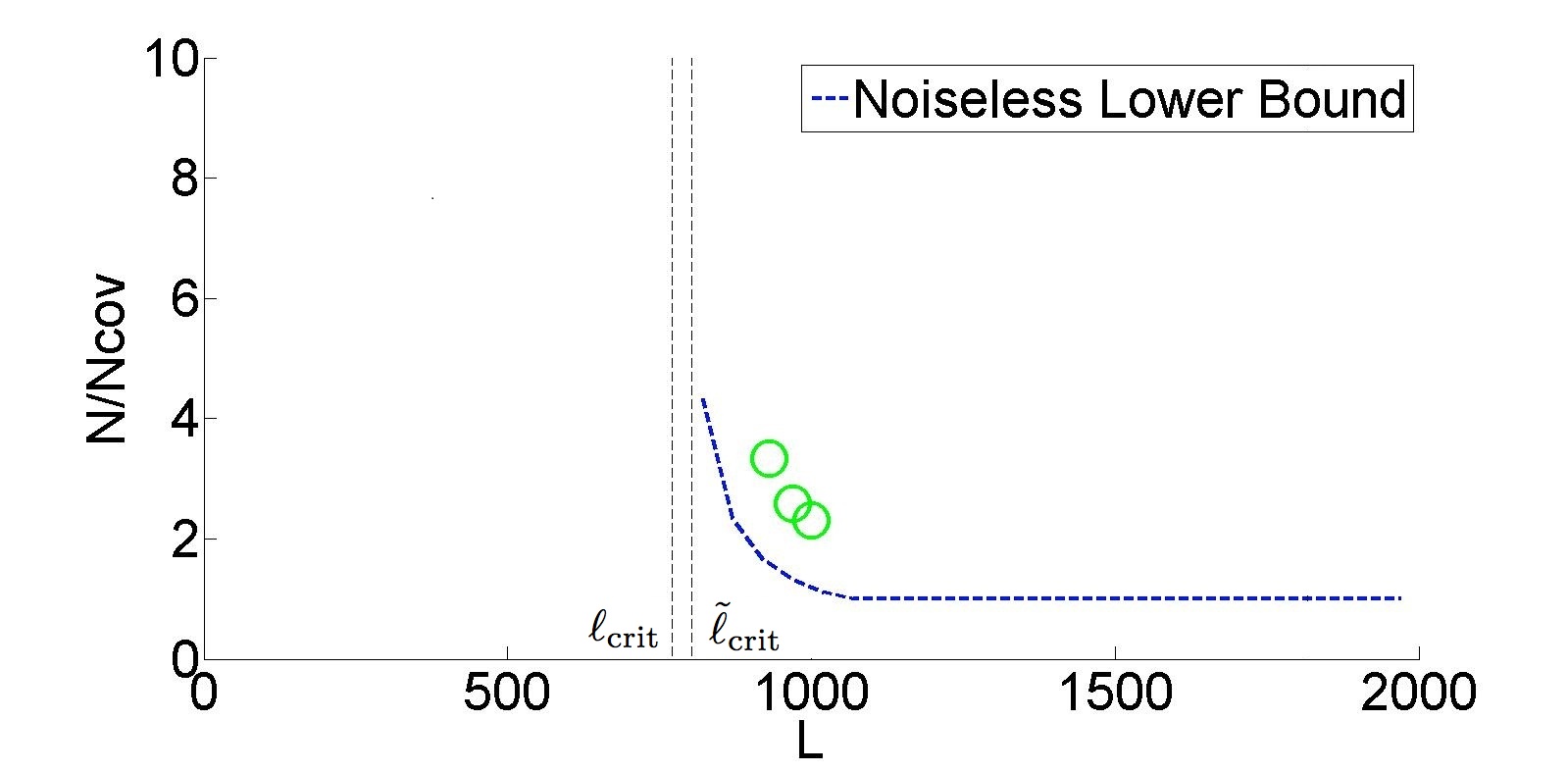}
}

\subfloat[\textit{Helicobacter pylori}]{\includegraphics[width=8cm]{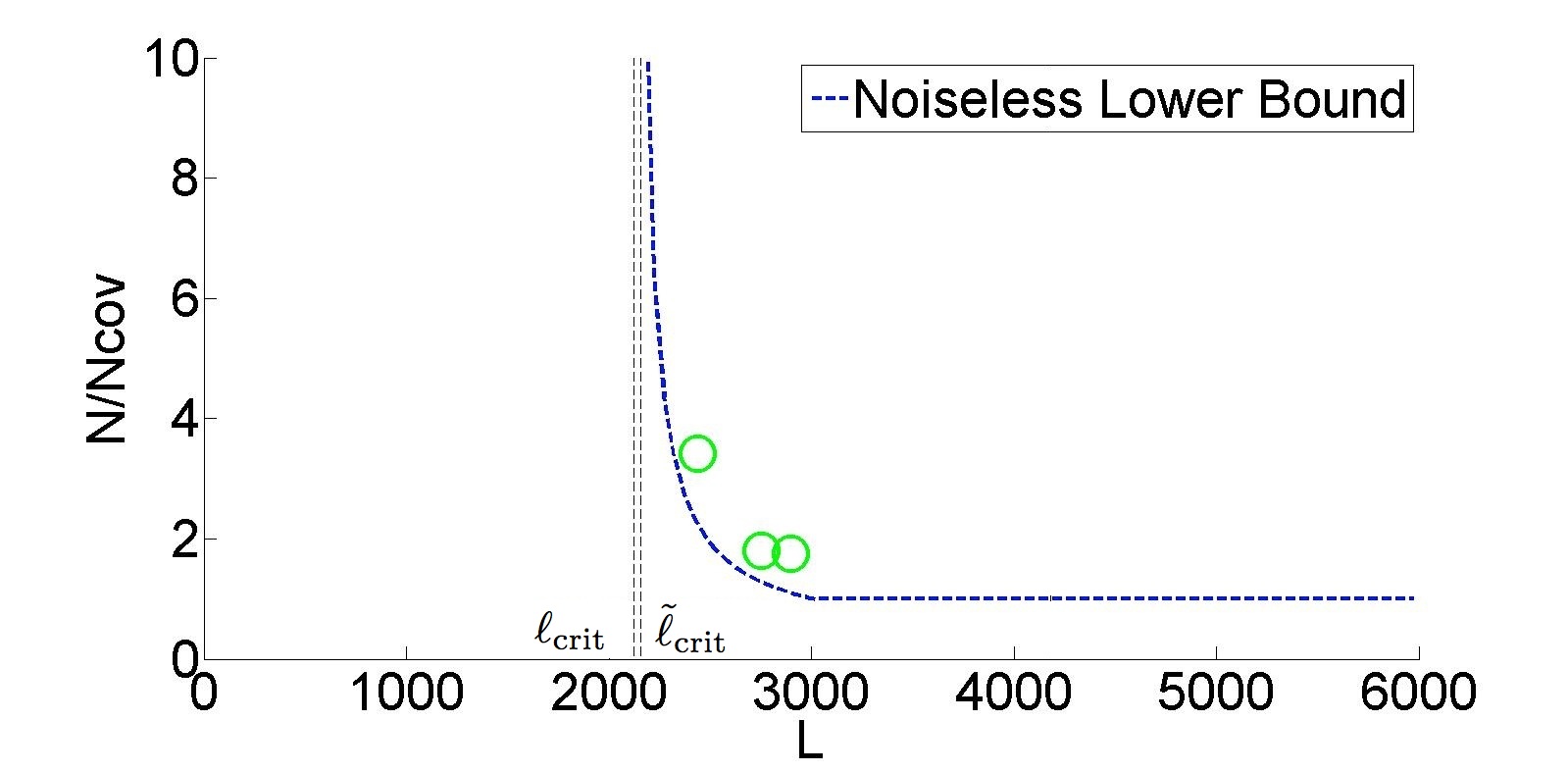}
}

\subfloat[\textit{Methanococcus maripaludis}]{\includegraphics[width=8cm]{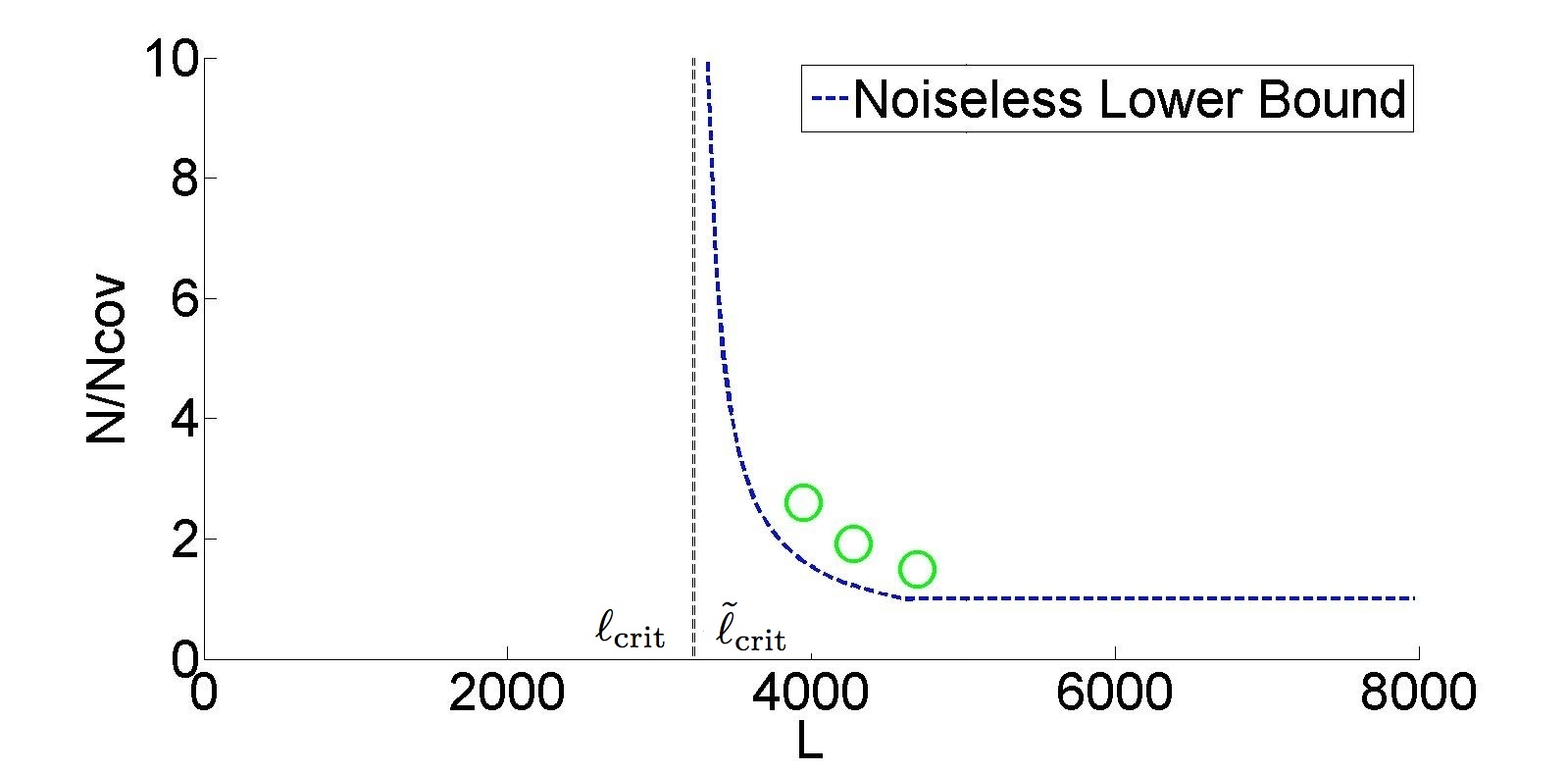}
}

\subfloat[\textit{Mycoplasma agalactiae}]{\includegraphics[width=8cm]{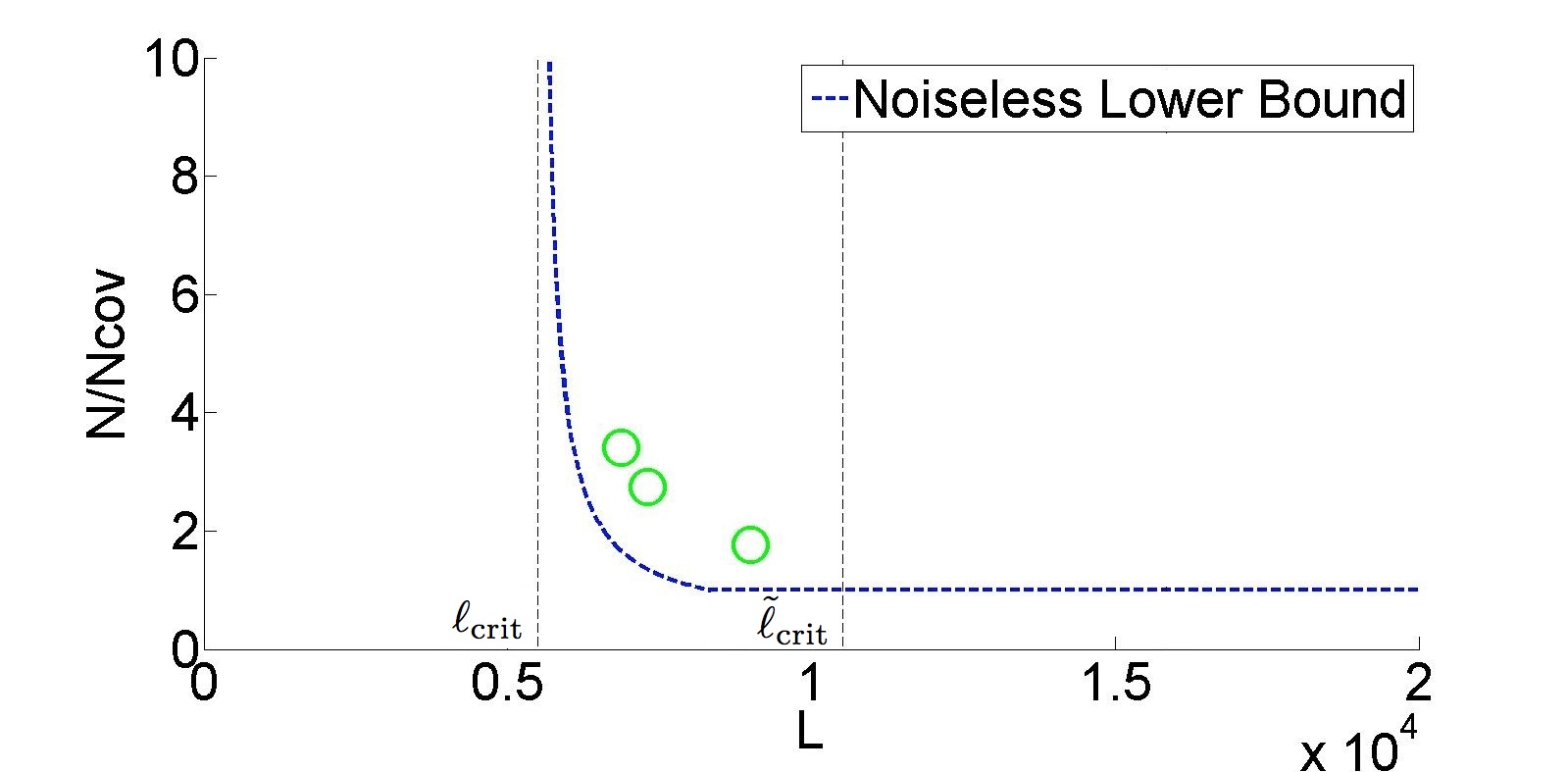}
}
\par
\begin{minipage}[t]{7cm}
\caption{\label{fig:Sim}Simulation results on a prototype  assembler (substitution noise of rate 1.5 \%) }
\end{minipage}
\end{figure}

The main conclusion of this work is that, with an appropriately designed assembly algorithm, the information requirement for genome assembly  is surprisingly insensitive to read noise.  The basic reason is that the redundancy required by the Lander-Waterman coverage constraint can be used to denoise the data.  This is consistent with the asymptotic result obtained in \cite{motahari2013optimal} and the practical approach taken in \cite{chin2013nonhybrid}. However, the result in \cite{motahari2013optimal} is based on a very simplistic i.i.d. random genome model, while the model and genomes considered in the present paper both have long repeats.
A natural extension of the Multibridging Algorithm in \cite{bresler2013optimal} to handle noisy reads allows the resolution of these long flanked repeats if the reads are long enough to span them, thus allowing reconstruction provided that the read length is greater than $L= \tilde{\ell}_{\rm crit} = \max \{\tilde{\ell}_{\rm int}, \tilde{\ell}_{\rm tri}\}$, where $\tilde{\ell}_{\rm int}$ is the length of the longest pair of flanked interleaved repeats and $\tilde{\ell}_{\rm tri}$  is the length of the longest flanked triple repeat in the genome. This condition is shown as a vertical asymptote of the "Multibridging Algorithm" curve in Figure \ref{noisyPlot}. By exploiting the redundancy in the read coverage to resolve read errors, the X-phased Multibridging can phase the polymorphism across the flanked repeat copies using only reads that span the exact repeats. Hence, reconstruction is achievable with a read length  close to $L = \ell_{\rm crit}$, which is the noiseless limit.



%

%
%

%
%
%
%


\subsection*{Related work}
All assemblers must somehow address the problem of resolving noise in the reads during genome reconstruction. However, the traditional approaches to measuring assembly performance makes quantitative comparisons challenging for unfinished genomes~\cite{NM2011}.  In most cases, the heart of the assembly problem lies in processing of the assembly graph, as in  \cite{zerbino2008velvet,gnerre2011high,simpson2012efficient}.  A common strategy for dealing with ambiguity from the reads lies in filtering the massively parallel sequencing data using the graph structure prior to traversing possible assembly solutions.  In the present work, however, we are focused on the often-overlooked goal of optimal data efficiency.  Thus, to the extent possible we distinguish between the read error and the mapping ambiguity associated with the shotgun sampling process.  
The proposed assembler, X-phased Multibridging, adds information to the assembly graph based on a novel analysis of the underlying reads.

\section*{Methods}

The path towards developing X-phased Multibridging is outlined as follows. 
%
%
\begin{enumerate}

\item Setting up the shotgun sequencing model and problem formulation. 

\item Analyzing repeats structure of genome and their relationship to the information requirement for genome finishing. 
\item Developing a parametric probabilistic model that captures the long tail of the repeat statistics. 
\item Deriving and analyzing an algorithm that require minimal information requirements for assembly -- close to the noiseless lower bound. 
\item Performing simulation-based experiments on real and synthetic genomes to characterize the performance of a prototype assembler for genome finishing.
\item Extending the algorithm to address the problem of indel noise.
  
\end{enumerate}

\section*{Shotgun sequencing model and problem formulation}

\subsection*{Sequencing model}

Let $\bf{s}$ be a length $G$ target genome being sequenced
with each base in the alphabet set $\Sigma=\{A,C,G,T\}$. In the
shotgun sequencing process, the sequencing instrument samples $N$ reads,
$\vec{r}_1, \ldots, \vec{r}_N$ of length $L$ and sampled uniformly and independently from  $\bf{s}$.  This unbiased sampling assumption is made for simplicity and is also supported by the characteristics of single-molecule (e.g. PacBio\textsuperscript{\textregistered}) data.  Each read is a noisy version of the corresponding length $L$ substring on the genome. The noise may consist of base insertions, substitutions or deletions.  Our analysis focus on substitution noise first. In a later section, indel noise is addressed. In the substitution noise model, let $p$ be the probability that a base is substituted by another base, with probability $p/3$ to be any other base. The errors are assumed to be independent across bases and across reads.

%
%

\subsection*{Formulation}
Successful reconstruction by an algorithm is defined by 
the requirement that, with probability at least $1-\epsilon$, 
the reconstruction $\hat{\bf{s}}$ is a single contig which is within edit distance $\delta$ from the 
target genome $\bf{s}$. 
If an algorithm can achieve that guarantee at some $(N,L)$, it is called $\epsilon$-feasible
at $(N,L)$.
This formulation implies automated genome finishing, 
because the output of the algorithm is one single contig.  
The fundamental limit for the assembly problem is the set of $(N,L)$  for 
which successful reconstruction is possible by some algorithms. 
If $\hat{\bf{s}}$ is directly spelled out from a correct placement of the reads, 
the edit distance between $\hat{\bf{s}}$ and $\bf{s}$ is of the order of $pG$, 
where the error rate is $p$. 
This motivates fixing $\delta = 2 p G$ for concreteness.
The quality of the assembly can be further improved if we follow the assembly 
algorithm with a consensus stage in which we correct each base, 
e.g. with majority voting. 
But the consensus stage is not the focus in this paper.
\section*{Repeats structure and their relationship to the information requirement for successful reconstruction}
\subsection*{Long exact repeats and their relationship to assembly with noiseless reads}
We take a moment to carefully define the various types of exact repeats. 
Let $\bf{s}^{\ell}_t$ denote the length-$\ell$ substring of the DNA sequence $\bf{s}$ starting at position $t$. 
An exact repeat of length $\ell$ is a substring appearing twice, at some positions $t_1, t_2$ (so $\bf{s}^{\ell}_{t_1}=\bf{s}^{\ell}_{t_2}$) that is maximal (i.e. $s(t_1 - 1) \neq  s(t_2 - 1)$ and $s(t_1 + \ell) \neq s(t_2 + \ell)$). 

Similarly, an exact triple repeat of length-$\ell$ is a substring appearing three times, at positions $t_1, t_2, t_3,$ such that $\bf{s}^{\ell}_{t_1}=\bf{s}^{\ell}_{t_2}=\bf{s}^{\ell}_{t_3}$, and such that neither of $s(t_1 - 1) = s(t_2 - 1) = s(t_3 - 1)$ nor $s(t_1 + \ell) = s(t_2 + \ell) = s(t_3 + \ell)$ holds. 

A copy of a repeat is a single one of the instances of the substring appearances. A pair of exact repeats refers to two exact repeats, each having two copies. 
A pair of exact repeats, one at positions $t_1, t_3$ with $t_1 < t_3$ and the second at positions $t_2, t_4$ with $t_2 < t_4$, is interleaved if $t_1 < t_2 < t_3 < t_4$ or $t_2 < t_1 < t_4 < t_3$. 
The length of a pair of exact interleaved repeats is defined to be the length of the shorter of the two exact repeats.
A typical appearance of a pair of exact interleaved repeat is --X--Y--X--Y-- where X and Y represent two different exact repeat copies and the dashes represent non-identical sequence content. 

We let $\ell_{\rm max}$ be the length of the longest exact repeat, $\ell_{\rm int}$ be the length of the longest pair of exact interleaved repeats and $\ell_{\rm tri}$ be the length of the longest exact triple repeat.

As mentioned in the introduction, it was observed that the read length and coverage depth required for successful reconstruction using noiseless reads for many genomes
is governed by long exact repeats. 
For some algorithms (e.g. Greedy Algorithm), the read length requirement is bottlenecked by $\ell_{\rm max}$.
The Multibridging Algorithm in \cite{bresler2013optimal} can successfully reconstruct the genome with a minimum amount of information.
The corresponding minimum read length requirement is the critical exact repeat length $\ell_{\rm crit}=\max(\ell_{\rm int},\ell_{\rm tri})$.
\subsection*{Flanked repeats}
While exact repeats are defined as the segments terminated on each end by a single differing base (Fig \ref{exact}), flanked repeats are defined by the segments terminated on each end by a statistically uncorrelated region. 
We call that ending region to be the \textit{random flanking region}. 
A distinguishing characteristic of the random flanking region is a high Hamming distance to segment length ratio between the ends of two repeat copies.
The ratio in the random flanking region is around 0.75, which matches with that when the genomic content is independently and uniformly randomly generated. 
We observe that long repeats of many genomes terminate with random flanking region.
Additional statistical analysis is detailed in the Appendix.

If the repeat interior is exactly the same between two copies of the flanked repeat (Fig \ref{nonpolymorphic}), the corresponding flanked repeat is called a flanked exact repeat. 
If there are a few edits (called polymorphism) within the repeat interior (Fig \ref{polymorphic}), the corresponding flanked repeat is called a flanked approximate repeat.

The length of the repeat interior bounded by the random flanking region is then the flanked repeat length. 
%
We let $\tilde{\ell}_{\rm max}$ be the length of the longest flanked repeat, $\tilde{\ell}_{\rm int}$ be the length of the longest pair of flanked interleaved repeats and $\tilde{\ell}_{\rm tri}$ be the length of the longest flanked triple repeat. The critical flanked repeat length is then $\tilde{\ell}_{\rm crit}=\max(\tilde{\ell}_{\rm int},\tilde{\ell}_{\rm tri})$.

\begin{figure} 
\subfloat[Exact repeat]{\includegraphics[width=8cm]{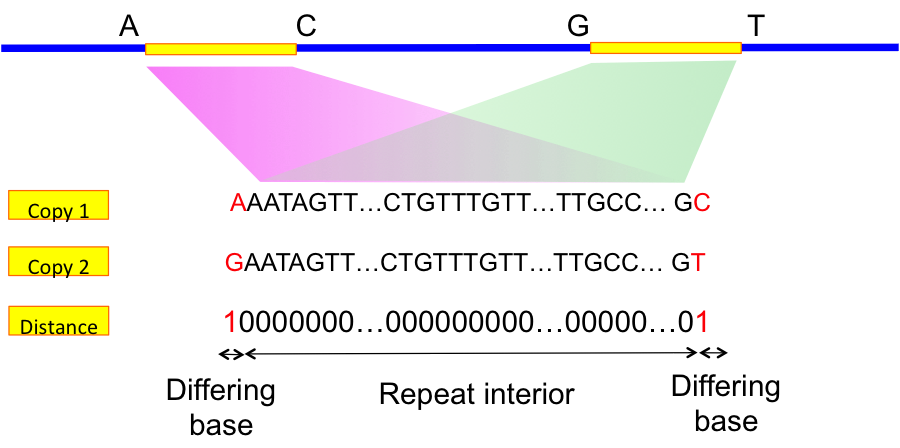}\label{exact}}

\subfloat[Flanked exact repeat]{\includegraphics[width=8cm]{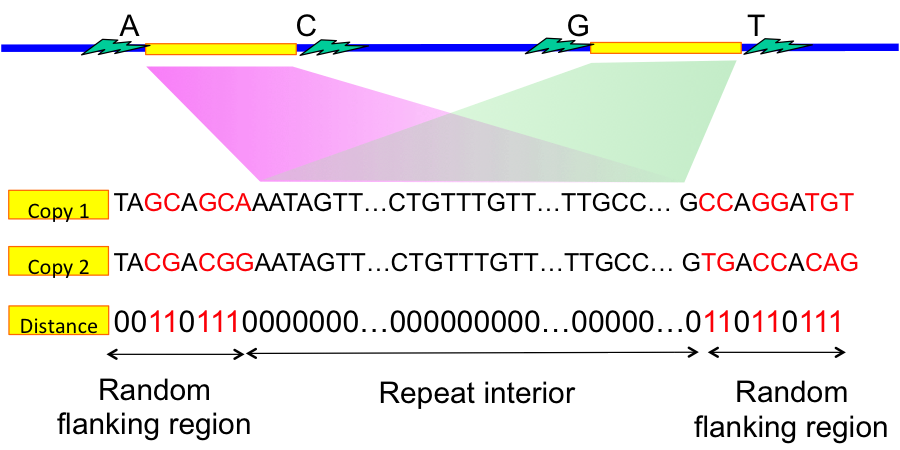}\label{nonpolymorphic}}

\subfloat[Flanked approximate repeat]{\includegraphics[width=8cm]{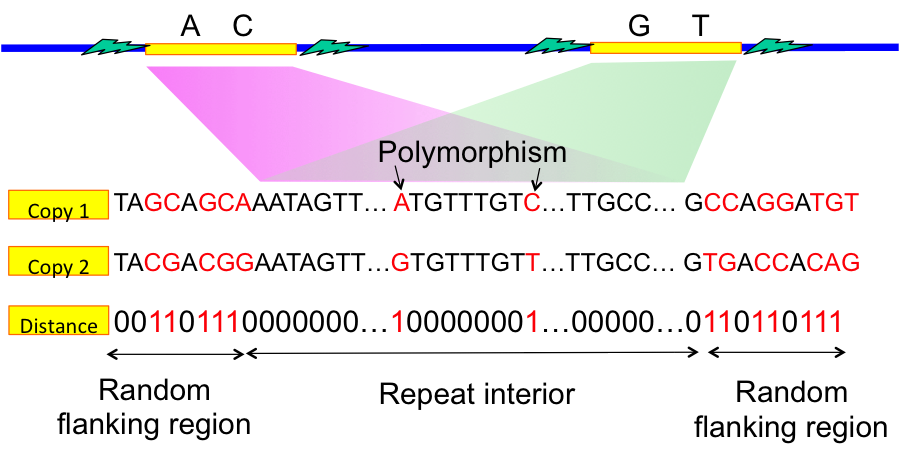}\label{polymorphic}
}
\par
\begin{minipage}[t]{7cm}
\caption{Repeat pattern}
\end{minipage}
\end{figure}

\begin{figure} 
  
\begin{minipage}[t]{7cm}
\subfloat[Noiseless reads spanning an exact repeat and its terminating bases]{
\includegraphics[width=7cm]{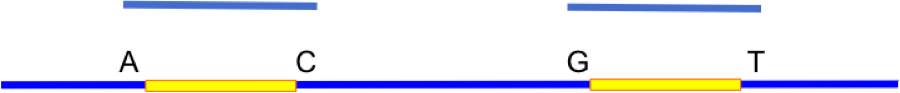}\label{intuition_1}}
\vspace{10pt}
  \end{minipage}
 \par

\begin{minipage}[t]{7cm}
\subfloat[Noisy reads spanning a flanked exact repeat  and its terminating random flanking region]{
\includegraphics[width=7cm]{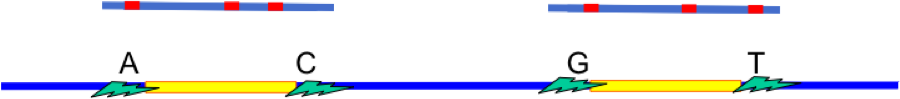}\label{intuition_2}}
\vspace{10pt}
\end{minipage}
 \par

\begin{minipage}[t]{7cm}
\subfloat[Noisy reads extending to span a flanked approximate repeat and its terminating random flanking region]{
\includegraphics[width=7cm]{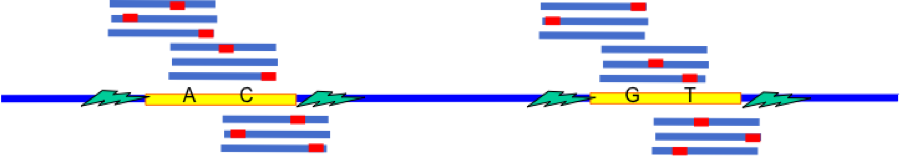}\label{intuition_3}
\vspace{10pt}
}
\end{minipage}
\par

\vspace{10pt}
\begin{minipage}[t]{7cm}
\caption{Intuition behind the information requirement}
\end{minipage}
\end{figure}

\subsection*{Long flanked exact repeats and their relationship to assembly with noisy reads}

If all long flanked repeats are flanked exact repeats, we 
can utilize the information in the random flanking region to generalize Greedy Algorithm and Multibridging Algorithm to handle noisy reads.
The corresponding information requirement is very similar to that when we are 
dealing with noiseless reads. 

The key intuition is as follows. 
A criterion for successful reconstruction is the existence of reads to span the repeats to their neighborhood.
When a read is noiseless, it only need to be long enough to span the repeat interior to its neighborhood 
by one base (Fig \ref{intuition_1}) so as to differentiate between two exact repeat copies. 
When a read is noisy, it then need to be long enough to span the repeat interior plus a short extension into 
the random flanking region (Fig \ref{intuition_2}) so as to confidently differentiate between two flanked repeat copies. 
However, the Hamming distance between two flanked repeat copies' neighborhood in the random flanking region is very high even within a short length. 
This can be used to differentiate between two flanked repeat copies confidently even when the reads are noisy. 
The short extension into the random flanking region has a length which is typically of order of tens whereas the long repeat length is of order of thousands.
Therefore, relative to the repeat length, the change of the critical read length requirement from handling noiseless reads to noisy reads is only marginal when all long repeats are flanked exact repeats.

\subsection*{Long flanked approximate repeats and their relationship to assembly with noisy reads}
%
If a long flanked repeat is a flanked approximate repeat, the flanked repeat length may be significantly longer than the length of its longest enclosed exact repeat. 
Merely relying on the information 
provided by the random flanking region requires the reads to be of length longer than 
the flanked repeat length for successful reconstruction. 
This explains why the information requirement for Greedy Algorithm and Multibridging Algorithm
has a significant increase when we use noisy reads instead of noiseless reads (as shown in Fig \ref{noisyPlot}).
However, if we utilize the information provided by the coverage, we can still confidently differentiate different 
repeat copies by phasing the small edits within the repeat interior (Fig \ref{intuition_3}). 
Specifically, we design X-phased Multibridging whose information requirement is close to the noiseless lower bound even when some long repeats are flanked approximate repeats, as shown in Fig \ref{noisyPlot}. 

\subsection*{From information theoretic insight to algorithm design}
Because of the structure of long flanked repeats,  there are two important sources of information 
that we specifically want to utilize when designing data-efficient algorithms to assemble noisy reads. 
They are 
\begin{itemize}
 \item The random flanking region beyond the repeat interior
 \item The coverage given by multiple reads overlapping at the same site
\end{itemize}

Greedy Algorithm(Alg \ref{alg:Noisy-Greedy}) utilizes the random flanking region when considering overlap. 
The minimum read length needed for successful reconstruction is close to $\tilde{\ell}_{\rm max}$. 

Multibridging Algorithm(Alg \ref{alg:Noisy-Multi-Bridging-1}) also utilizes the random flanking region but it improves upon Greedy Algorithm
 by using a De Bruijn graph to aid the resolution of flanked repeats. 
The minimum read length needed for successful reconstruction is close to $\tilde{\ell}_{\rm crit}$. 

X-phased Multibridging(Alg \ref{alg:Alignment-bridging}) further utilizes the coverage given by multiple reads 
to phase the polymorphism within the repeat interior of flanked approximate repeats. 
The minimum read length needed for successful reconstruction is close to $\ell_{\rm crit}$, which is the noiseless lower bound even when some long repeats are flanked approximate repeats.

\section*{Model for genome}
To capture the key characteristics of repeats and to guide the design of assembly algorithms, 
we use the following parametric probabilistic model for genome. 
A target genome is modeled as  a random vector $\bf{s}$ of length
$G$ that has the following three key components (a pictorial representation is depicted in 
Figure \ref{fig:Generative-model-for}).  

\textbf{Random background:}
The background of the genome is a random vector, composed of uniformly and 
independently picked bases from the alphabet set $\Sigma=\{A,C,G,T\}$. 

\textbf{Long flanked repeats:}
On top of the random background, we randomly position the longest flanked repeat and the longest flanked triple repeat. Moreover, we randomly position a flanked repeat interleaving the longest flanked repeat, forming the longest pair of flanked interleaved repeat.
The corresponding length of the flanked repeats are  $\tilde{\ell}_{\rm max}$, 
$\tilde{\ell}_{\rm tri}$ and $\tilde{\ell}_{\rm int}$ respectively. 
It is noted that  $ \tilde{\ell}_{\rm max} > \max (\tilde{\ell}_{\rm int},\tilde{\ell}_{\rm tri})$.

\textbf{Polymorphism and long exact repeats:}
Within the  repeat interior of the flanked repeats, 
we randomly position $n_{\rm max}$, $n_{\rm int}$  and $n_{\rm tri}$ edits (polymorphism) respectively. 
The sites of polymorphism are chosen such that the longest exact repeat, 
the longest pair of exact interleaved repeats and the longest exact triple repeat are of 
length $\ell_{\rm max}$, $\ell_{\rm int}$ and $\ell_{\rm tri}$ respectively.

\begin{figure} 
\includegraphics[width=8cm]{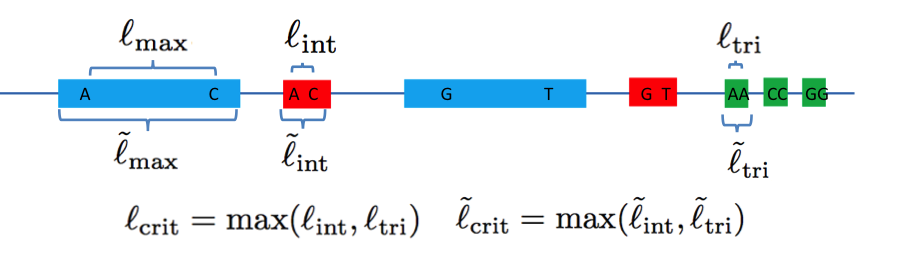}
\caption{Model for genome \label{fig:Generative-model-for}}
\end{figure}

\section*{Algorithm design and analysis}

\subsection*{Greedy Algorithm}
Read $R_2$ is a successor of read $R_1$ if there exists length-$W$ suffix of $R_1$ 
and length-$W$ prefix of $R_2$ such that they are extracted from the same 
locus on the genome. 
Furthermore, there is no other reads that can satisfy the same condition with a
larger $W$. 
To properly determine successors of reads in the presence of long repeats and noise, we need to define an appropriate overlap rule for reads. In this section, we show the conceptual development towards defining such a rule, which is called RA-rule.

\textbf{Noiseless reads and long exact repeats: }If the reads are noiseless, 
all reads can be paired up with their successors correctly with high probability 
when the read length exceeds $\ell_{\rm max}$.  
It was done \cite{bresler2013optimal} by greedily pairing reads and their candidate successors based on their overlap score in descending 
order. 
When a read and a candidate successor are paired, they will be removed from the 
pool for pairing.
Here the overlap score between a read and a candidate successor is the maximum length 
such that the suffix of the read and prefix of the candidate successor match \textit{exactly}. 

\textbf{Noisy reads and random background: }Since we cannot expect exact match for noisy reads, we need a different way to 
define the overlap score. 
Let us consider the following toy situation. 
Assume that we have exactly one length-$(\ell +1) $ noisy read starting at each locus of a 
length $G$ random genome(i.e. only consists of the random background). 
Each read then overlaps with its successor precisely by $\ell$ bases. 
Analogous to the noiseless case, one would expect to pair reads greedily based 
on overlap score. 
Here the overlap score between a read and a candidate successor is the maximum length 
such that the suffix ($x$) of the read and prefix  ($y$) of the candidate successor match \textit{approximately}. 
To determine whether they match \textit{approximately}, one can use a  predefined a threshold factor 
$\alpha$ and compute the  Hamming distance  $d(x,y)$.
If $d(x,y)\le \alpha\cdot \ell$, then they match \textit{approximately}, otherwise not. 
Given this decision rule, we can have false positive (i.e. having any pairs of reads mistakenly paired up) and false negative
(i.e. having any reads not paired up with the true successors). 
If false positive and false negative probability are small, this naive method is a reliable 
enough metric.
This can be achieved by using a long enough length $\ell>\ell_{\rm iid}$ 
and an appropriately chosen threshold $\alpha$. 

Recall that $\epsilon$ is the overall failure probability.
By bounding the sum of false positive and false negative probability by $\epsilon/3$, 
one can find $\ell_{\rm iid}(p,\epsilon/3,G)$ and $\alpha(p,\epsilon/3,G)$ to be the $(\ell_{\rm iid},\alpha)$ solution to the following pair of equations: 
\begin{eqnarray}
G^{2}\cdot  \exp(-\ell_{\rm iid}\cdot D(\alpha||\frac{3}{4})) &=\frac{\epsilon}{6}\\
G\cdot  \exp(-\ell_{\rm iid}\cdot D(\alpha||2p-\frac{4}{3}p^{2}))&=\frac{\epsilon}{6}
\end{eqnarray}
where $D(a||b)=a\log\frac{a}{b}+(1-a)\log\frac{1-a}{1-b}$ is the Kullback-Leibler divergence.

\textbf{Noisy reads and long flanked repeats: }However, when the genome contains long flanked repeats on top of the random background, 
this naive rule of determining  overlap is not enough. 
Let us look at the example in Fig. \ref{fig:Intuition-about-why-1}.
\begin{figure}
\begin{centering}
\includegraphics[width=8cm]{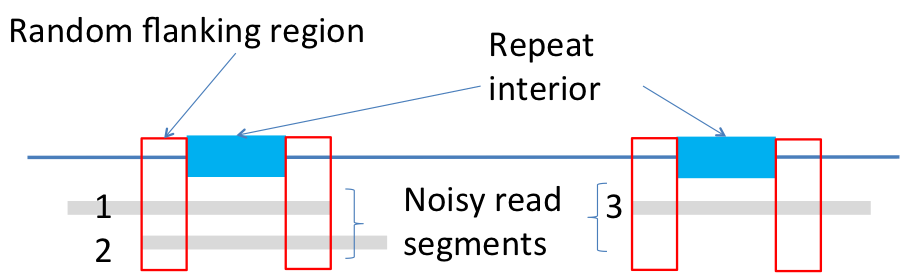} 
\par\end{centering}

\begin{minipage}[t]{7cm}
\caption{\label{fig:Intuition-about-why-1}Intuition about why we define the overlap rule to be RA-overlap rule}
\end{minipage}
\end{figure}
As shown in Fig. \ref{fig:Intuition-about-why-1}, because of long flanked repeats, we have a small ratio of overall distance against the overlap length for
 the segments that are extracted from different copies of the repeat (e.g Segment 1 and Segment 3 in Fig. \ref{fig:Intuition-about-why-1}). 
Therefore, the overall Hamming distance between two segments is not a good enough metric for defining overlap.
If we abide by the naive rule, we need to increase the read length significantly longer than the flanked repeat length so as to guarantee 
confidence in deciding approximate match. 
Otherwise, it will either result in a high false positive rate (if we set a large $\alpha$) or a high false negative 
rate (if we set a small $\alpha$). 
To properly handle such scenario, we define a repeat-aware rule(or RA-rule).
\begin{itemize}
\item RA-matching: Two segments $(x,y)$ of length $W$ match under the RA-rule  if and only 
if the distance between whole segments
is $<\alpha\cdot W$ and \textbf{both} of its ending segments(of length
$\ell_{\rm iid}$) also have distance $<\alpha\cdot \ell_{\rm iid}$.  
\item RA-overlap: The overlap score between a read and a candidate successor 
under the RA-rule is the maximum length 
such that the suffix of the read and prefix of the candidate successor match under the RA-matching. 
\end{itemize}
The RA-rule is particularly useful because it puts an emphasis on both ends of the overlap region. 
Since the ends are separated by a long range, one end will hopefully  
originate from the random flanking region of the flanked repeat.
If we focus on the segments originating from the random flanking region, the distance per segment length 
ratio will be very high when the segments originate from different copies of the repeat but 
very low when they originate from the same copy of the repeat.
This is how we utilize the random flanking region to differentiate between repeat copies and determine 
correct successors in the presence of long flanked repeats and noise. 

If we use  Greedy Algorithm (Alg \ref{alg:Noisy-Greedy}) to merge reads greedily with this
overlap rule (RA-rule), Prop \ref{prop:greedy} shows the information requirement under the previously described sequencing model and genome model. 
A plot is shown in Fig \ref{noisyPlot}.  
Since $\ell_{\rm iid}$ is of order of tens whereas $\tilde{\ell}_{\rm max}$ is 
of order of thousands, the read length requirement for Greedy Algorithm 
to succeed is dominated by $\tilde{\ell}_{\rm max}$. 
The detailed proof of Prop \ref{alg:Noisy-Greedy} is given in Appendix. 

\begin{algorithm}[tbh]
\raggedright{}
Initialize contigs to be reads \;

\For{ $W=L$ to $\ell_{\rm iid}$ } {
\If {any two contigs x,y are of overlap $W$ under RA-rule }
{  merge $x,y$ into one contig.}}

\caption{\label{alg:Noisy-Greedy}Greedy Algorithm}
\end{algorithm}

\begin{proposition}
\label{prop:greedy}With $\ell_{\rm iid}=\ell_{\rm iid}(p,\frac{\epsilon}{3},G)$,
if
$$L > \tilde{\ell}_{\rm max}+2\ell_{\rm iid},$$
\[
N \hspace{1mm} > \hspace{1mm} \max \left( \frac{G \ln (3/\epsilon) }{L - \tilde{\ell}_{\rm max} - 2 \ell_{\rm iid}}  , \hspace{1mm}
\frac{G\ln (3 N/\epsilon)}{L - 2 \ell_{\rm iid}} 
\right)
\]

then, Greedy Algorithm(Alg \ref{alg:Noisy-Greedy}) is $\epsilon-$feasible at $(N, L)$. 
\end{proposition}


\subsection*{Multibridging Algorithm}

The read length requirement of Greedy Algorithm has a bottleneck around 
$\tilde{\ell}_{\rm max}$ because it requires at least one copy of each flanked repeat to be spanned by at least one read for successful reconstruction.  
Spanning a repeat by a single read is called bridging in \cite{bresler2013optimal}. 
A natural question is whether we need to have all repeats bridged for successful reconstruction.

In the noiseless setting, \cite{bresler2013optimal} shows that this condition 
can be relaxed. 
Using noiseless reads, one can have successful reconstruction given all copies of each exact triple repeat being bridged, and at least one copy of one of the repeats in each pair of exact interleaved repeats being bridged.

A key idea to allow such a relaxation in \cite{bresler2013optimal}  is to use a De Bruijn graph to capture the structure of the genome.

When the reads are noisy, we can utilize the random flanking region to specify a De Bruijn graph 
with high confidence by RA-rule and arrive at a similar relaxation. 
By some graph operations to handle the residual errors, we can have successful reconstruction 
with read length $\tilde{\ell}_{\rm crit}+2\cdot \ell_{\rm iid}<L<\tilde{\ell}_{\rm max}$. 
The algorithm is summarized in Alg \ref{alg:Noisy-Multi-Bridging-1}.
Prop \ref{prop:NMB} shows its information requirement under the previously described sequencing model and genome model.
A plot is shown in Fig \ref{noisyPlot}. 
We note that Alg \ref{alg:Noisy-Multi-Bridging-1} can be 
seen as a noisy reads generalization of Multibridging Algorithm 
for noiseless reads in \cite{bresler2013optimal}. 

\subsubsection*{Description and its performance}
\begin{proposition}
\label{prop:NMB}With $\ell_{\rm iid}=\ell_{\rm iid}(p,\frac{\epsilon}{3},G)$,\textup{
if}

$$ L > \tilde{\ell}_{\rm crit}+2\ell_{\rm iid}, $$
\[
N \hspace{1mm} > \hspace{1mm} \max \left( \frac{G \ln (3/\epsilon) }{L - \tilde{\ell}_{\rm crit} - 2 \ell_{\rm iid}}  , \hspace{1mm}
\frac{G\ln (3 N/\epsilon)}{L - 2 \ell_{\rm iid}} 
\right)
\]
then,  Multibridging Algorithm(Alg \ref{alg:Noisy-Multi-Bridging-1})
is $\epsilon-$feasible at $(N, L)$. 
\end{proposition}

\begin{algorithm}[tbh]
\raggedright{}

1. Choose K to be $\tilde{\ell}_{\rm crit}+2\ell_{\rm iid}$ and extract K-mers
from reads.

2. Cluster K-mers based on the RA-rule.

3. Form uncondensed De Bruijn graph $G_{De-Bruijn}=(V,E)$ with the following rule:

\begin{itemize}
\item a) K-mers clusters as node set $V$.
\item b)  $(u,v)=e\in E$ if and only if there exists K-mers
$u_{1}\in u$ and $v_{1}\in v$ such that $u_{1}$,$v_{1}$are consecutive K-mers in some reads.
\end{itemize}

4. Join the disconnected components of $G_{De-Bruijn}$ together by
the following rule:

\For{ $W=K-1$ to $\ell_{\rm iid}$} {

\For {each node $u$ which has either no predecessors / successors
in $G_{De-Bruijn}$ }{

a) Find the predecessor/successor $v$ for $u$ from all possible
K-mers clusters such that overlap length(using any representative
K-mers in that cluster) between $u$ and $v$ is $W$ under RA-rule.

b) Add dummy nodes in the De Bruijn graph to link $u$ with $v$
and update the graph to $G_{De-Bruijn}$
}
}
5. Condense the graph $G_{De-Bruijn}$ to form $G_{string}$ with
the following rule:

\begin{itemize}
\item a) Initialize $G_{string}$ to be $G_{De-Bruijn}$ with node labels of
each node being its cluster group index.

\item b) \While {$\exists$successive nodes $u\to v$ such that $out-degree(u)=1$and
$in-degree(v)=1$}{

\hspace{10pt} bi) Merge $u$ and $v$ to form a new node $w$

\hspace{10pt} bii) Update the node label of $w$ to be the concatenation of node labels of $u$ and $v$
}
\end{itemize}

6. Clear Branches of $G_{string}$:

\For {each node $u$ in the condensed graph $G_{string}$}{

\If {$out-degree(u)>1$ and that all the successive paths are of the
same length(measured by the number of node labels) and then joining
back to node $v$ and the path length $<\ell_{\rm iid}$}
{we merge the paths into a single path from $u$ to $v$. }

}
7. Condense graph $G_{string}$

8. Find the genome :

\begin{itemize}
\item a) Find an Euler Cycle/Path in $G_{string}$ and output the concatenation
of the node labels to form a string $\vec{\bf{s}}_{labels}$.

\item  b) Using $\vec{\bf{s}}_{labels}$ and look up the associated K-mers to
form the final recovered genome $\hat{\bf{s}}$. 
\end{itemize}

\caption{\label{alg:Noisy-Multi-Bridging-1}Multibridging Algorithm}
\end{algorithm}

Detailed proof is given in the Appendix. The following sketch highlights the motivation behind the key steps of Multibridging Algorithm.

{[}Step1{]} We set a large K value to make sure the K-mers overlapping
the shorter repeat of the longest pair of flanked interleaved repeats and the longest flanked
triple repeat can be separated as distinct clusters.

{[}Step2{]} Clustering is done using the RA-rule because of the existence
of long flanked repeats and noise.

{[}Step3{]} A K-mer cluster corresponds to an equivalence class for K-mers matched under the RA-rule. This step forms a De Bruijn graph with K-mer clusters as nodes.

{[}Step4{]} Because of large K, the graph can be disconnected due
to insufficient coverage. In order to reduce the coverage constraint,
we connect the clusters greedily.

{[}Step5, 7{]} These two steps simplify the graph.

{[}Step6{]} Branch clearing repairs any incorrect merges near the boundary of long flanked repeat. 

{[}Step8{]} Since an Euler path in the condensed graph corresponds to the correct genome sequence, it is traversed to form the reconstructed genome.

\subsubsection*{Some implementation details: improvement on time and space efficiency }
For Multibridging Algorithm, the most computational expensive step is the clustering of K-mers. 
To improve the time and space efficiency, this clustering step can be approximated by performing pairwise comparison of reads.

Based on the alignment of the reads, we can cluster K-mers from different reads together 
using a disjoint set data structure that supports union and find operations.
Since only reads are used in the alignment, only the K-mer indices along with their associated read indices and offsets need to be stored in memory--- not all the K-mers. 

Pairwise comparison of reads roughly runs in $\tilde{\Theta}(N^{2}L^{2})$ if done in the naive way.
To speed up the pairwise comparison of noisy reads, one can utilize the fact that the read length is 
long.
We can extract all consecutive $f$-mers (which act as fingerprints) of the reads and do a lexicographical sort to find candidate neighboring reads and associated offsets for comparison. 
Since the reads are long, if two reads overlap, there should exist some perfectly matched $f$-mers which can be identified after the lexicographical sort.
This allows an optimized version of Multibridging Algorithm to run in $\tilde{\Theta}(NL\cdot \frac{NL}{G})$ time and $\tilde{\Theta}(NLf)$ space.

\subsection*{X-phased Multibridging}

As shown in Fig \ref{noisyPlot}, when long repeats are flanked approximate repeats, there
can be a big gap between the noiseless lower bound and the information
requirement for Multibridging Algorithm. 
A natural question is whether this is due to a fundamental lack of information from the 
reads or whether Multibridging Algorithm does not utilize all the available information.
In this section, we demonstrate that there is an important source of information
 provided by coverage which is not utilized by Multibridging Algorithm. 
In particular, we introduce  X-phased Multibridging, an assembly algorithm that utilizes the information 
provided by coverage to phase the polymorphism in long flanked repeat interior.
The information requirement of X-phased Multibridging is close to the noiseless 
lower bound (as shown in Fig \ref{noisyPlot}) even when some long repeats are flanked approximate repeats.

  
\subsubsection*{Description of  X-phased Multibridging}
\begin{figure}
\begin{centering}
\subfloat[Consensus Step ]{\includegraphics[width=8cm]{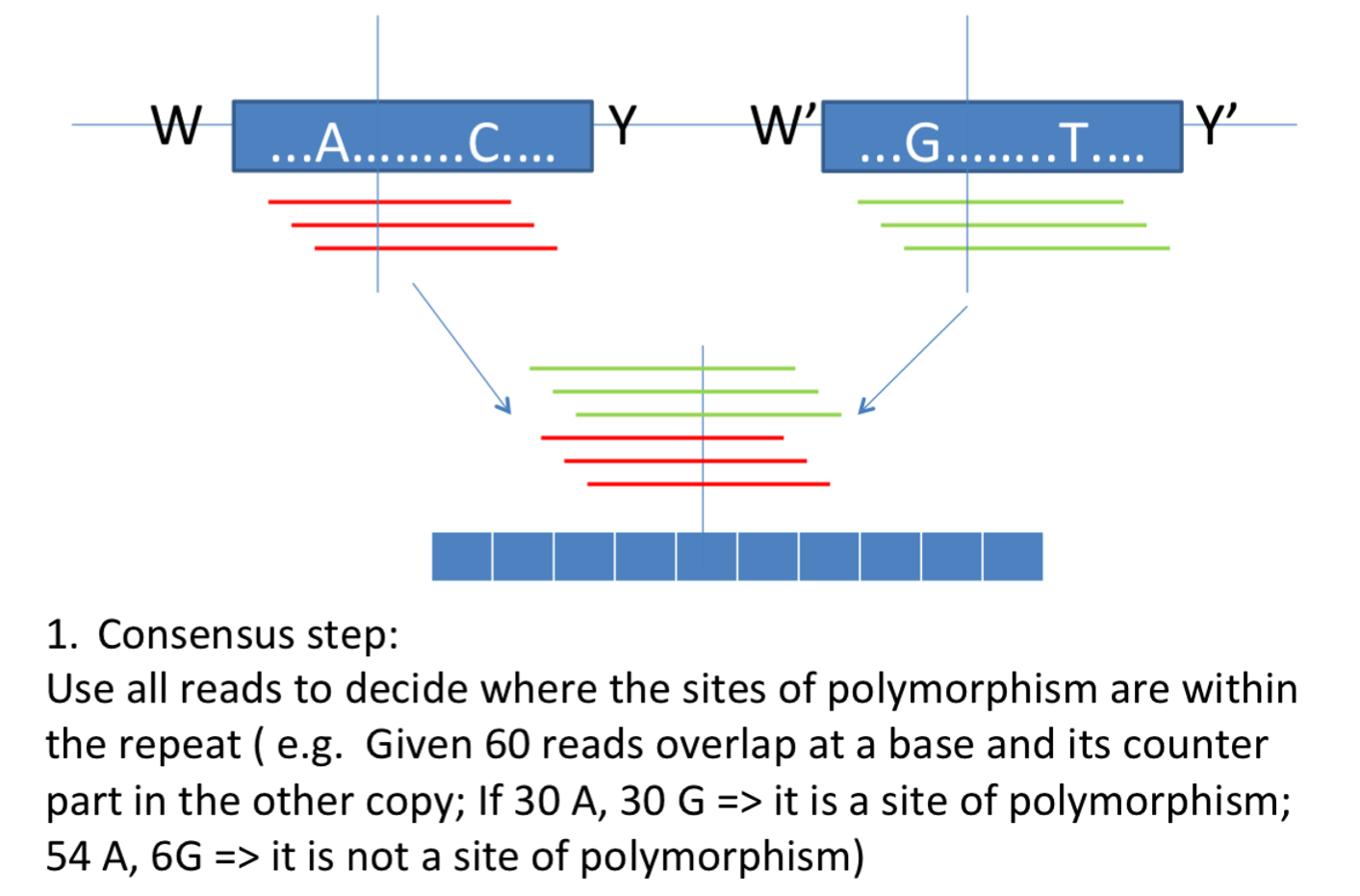}
}

\subfloat[Read Extension Step]{\includegraphics[width=8cm]{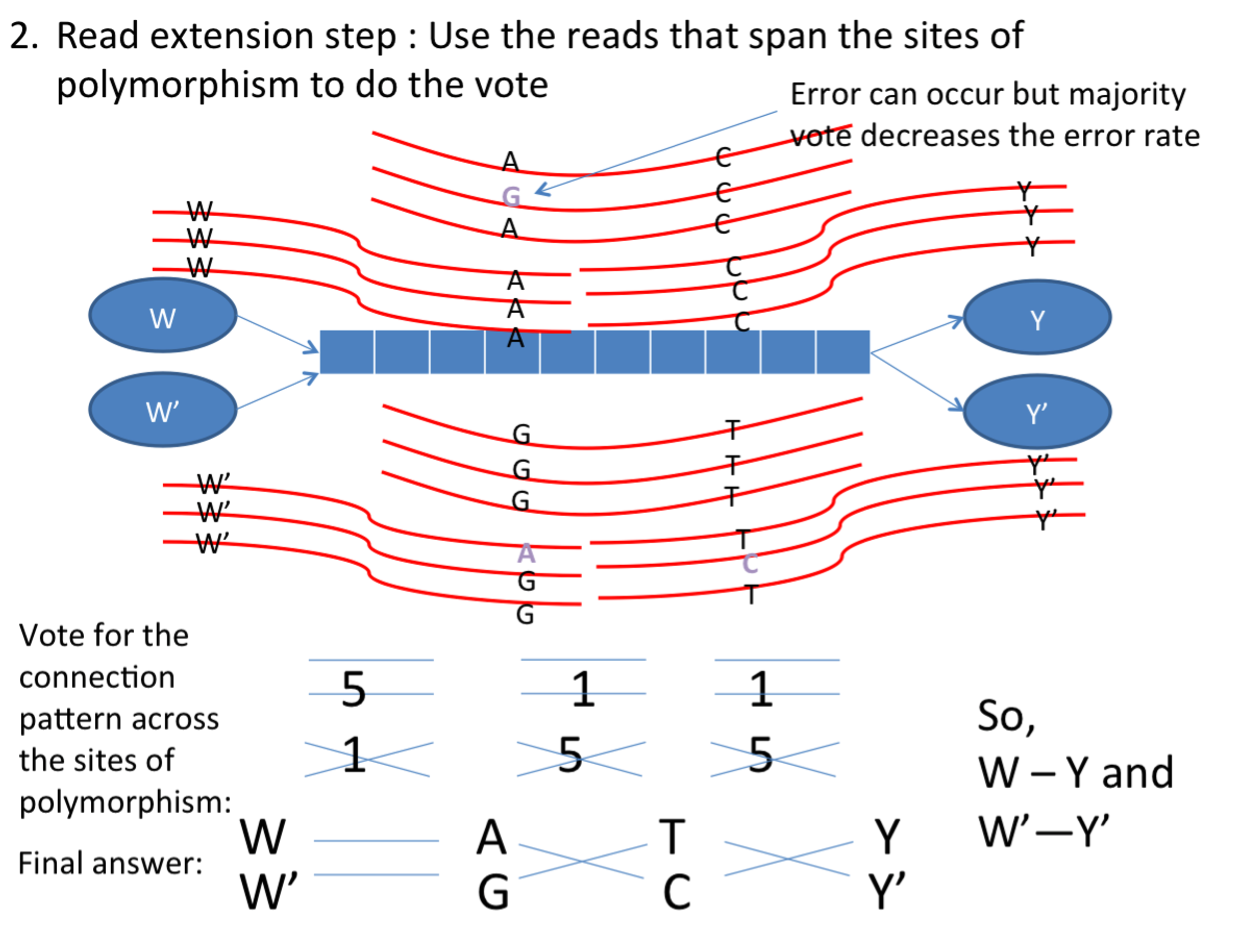}}
\par\end{centering}

\begin{minipage}[t]{7cm}
\caption{\label{fig:Illustration-of-how}Illustration of how to phase polymorphism to extend reads across repeats}
\end{minipage}
\end{figure}

Multibridging Algorithm utilizes the random flanking region to differentiate between repeat copies.
However, for a flanked approximate repeat, its enclosed exact repeat does not terminate with the 
random flanking region but only terminates with sparse polymorphism. 
When we consider the overlap of \textit{two} reads originating from different copies of a flanked approximate repeat, 
the distinguishing polymorphism is so sparse that it cannot be used to confidently differentiate between repeat copies. 
Therefore, there is a need to use the extra redundancy introduced by the 
coverage from \textit{multiple} reads to confidently differentiate between repeat copies and 
that is what X-phased Multibridging utilizes. 

X-phased Multibridging (Alg \ref{alg:Alignment-bridging}) follows the algorithmic design of Multibridging Algorithm. 
However, it adds an extra phasing procedure to differentiate between repeat copies of long flanked repeats that Multibridging Algorithm cannot confidently differentiate. 
We recall that after running step 7 of Multibridging Algorithm, a node in the graph $G_{string}$ corresponds to a 
substring of the genome and has node label consisting of consecutive K-mer cluster indices.
An X-node of $G_{string}$ is a node that has in-degree and out-degree $\ge 2$.
X-node indeed corresponds to a flanked repeat.
The incoming/outgoing nodes of the X-node correspond to the incoming/outgoing 
random flanking region of the flanked repeat. 

To be concrete, we focus the discussion on a pair of flanked interleaved repeats, assuming triple repeats
are not the bottleneck. However, the ideas presented can be generalized to repeats of more copies. 

For the flanked approximate repeat with length $\ell_{\rm int}<L$ 
and $\tilde{\ell}_{\rm int}>L$ (as shown in Fig \ref{fig:Illustration-of-how}),  there is no node-disjoint 
paths joining incoming/outgoing random flanking region with the distinct repeat copies in $G_{string}$.
It is because the reads are corrupted by noise and the polymorphism is too sparse to differentiate between the repeat copies. 
Executing Multibridging Algorithm directly will result in the formation of an X-node, which is an artifact due to K-mers from different copies of the flanked approximate repeat erroneously clustered together.

Successful reconstruction requires an algorithm to pair up the correct 
incoming/outgoing nodes of the X-node(i.e. decide how $W,W'$
and $Y,Y'$ are linked in Fig \ref{fig:Illustration-of-how}). 
This is handled by the phasing procedure in X-phased Multibridging, which uses all the reads information.
The phasing procedure is composed of two main steps:
\begin{itemize}
  \item Consensus step: Confidently find out where the sites of polymorphism are located 
  within the flanked repeat interior. 
  \item Read extension step: Confidently determine how to extend reads using the random flanking region 
  and sites of polymorphism as anchors. 
\end{itemize}

\textbf{Consensus step} 
For the X-node of interest,  let $D$ be the set of reads originating from any sites of the associated flanked repeat region and let $x_{1}$ and $x_{2}$ denote the associated repeat copies. 
Since the random flanking region is used as anchor, it is treated as 
the starting base (i.e. $x_{1}(0)=W$ and
$x_{2}(0)=W'$). 
For the $i^{th}$ subsequent site of the flanked repeat (where
$1\le i\le \tilde{\ell}_{\rm int}$),
we determine the consensus according to Eq (\ref{eq:consen}).  
This can be implemented by
counting the frequency of occurrence of each alphabet overlapping at each site of the 
repeat.
The consensus result determines the sites of polymorphism and the most likely pairs of bases 
at the sites of polymorphism.
\begin{equation}
\max_{F\subset\{A,C,G,T\}^{2}}{\cal P}(\{x_{1}(i),x_{2}(i)\}=F\mid D)\label{eq:consen}
\end{equation}

\textbf{Read extension step} After knowing the sites of polymorphism, we use those reads that 
span the sites of polymorphism or random flanking region to help decide how to extend reads 
across the flanked repeat.
Let $\sigma$ be
the possible configuration of alphabets at the sites of polymorphism
and random flanking region (e.g. $\sigma=(ACY,GTY')$ means that the two copies of the flanked repeat 
with the corresponding random flanking region respectively are W--A--C--Y, W'--G--T--Y' where 
the common bases are omitted). 

The following maximum a posteriori estimation is used to decide the correct configuration. 
\begin{equation}
\max_{\sigma}{\cal P}(\hat{\sigma}=\sigma\mid D,\{x_{1}(i),x_{2}(i)\}_{i=1}^{\tilde{\ell}_{\rm int}})\label{eq:ML}
\end{equation}
 where $\hat{\sigma} $ is the estimator, $D$ is the raw read set, and $x_1,x_2$ are the estimates from the consensus step.
It is noted that the size of the feasible set for $\sigma$ is $2^{n_{\rm int}+1}$.

In practice, for computational efficiency, the maximization in Eq (\ref{eq:ML}) can be approximated accurately 
even if it is replaced by the simple counting illustrated in Fig \ref{fig:Illustration-of-how}, which
we call count-to-extend algorithm(countAlg). 
CountAlg uses the raw reads to establish majority vote on how one should extend to the next sites 
of polymorphism using only the reads that span the sites of polymorphism.

\begin{algorithm}
\raggedright{}
1. Perform Step 1 to Step 7 of MultiBridging Algorithm

2. For every X-node $x\in G_{string}$
\begin{itemize}
\item a) Align all the relevant reads to the flanked repeat $x$
\item b) Consensus step: Consensus to find location of polymorphism by solving Eq (\ref{eq:consen})
\item c) Read extension step: If possible, resolve flanked repeat(i.e. pair up the incoming/outgoing nodes of $x$) by either countAlg or by solving Eq (\ref{eq:ML}) 
\end{itemize}

3. Perform Step 8 of MultiBridging Algorithm as in Alg \ref{alg:Noisy-Multi-Bridging-1}

\caption{\label{alg:Alignment-bridging} X-phased Multibridging}
\end{algorithm}

\subsubsection*{Performance }
After introducing the phasing procedure in X-phased Multibridging, we proceed to 
find its information requirement for successful reconstruction. 

The information requirement for X-phased Multibridging is the amount of information required to reduce the error of the phasing procedure to a negligible level.  The phasing procedure -- step 2 in Alg. \ref{alg:Alignment-bridging} -- is a combination of consensus and read extension steps, which contribute to the error as follows.
 
Let ${\cal E}$ be the error event of the repeat phasing procedure for a repeat,
$\epsilon_{1}$
be the error probability for the consensus step,
$\epsilon_{2}$ be
the error probability for the read extension step given $k$ reads
spanning each consecutive site of polymorphism within the flanked repeat,
$\delta_{cov}$ be the probability for having $k$ reads spanning
each consecutive sites of polymorphism(i.e. $k$ bridging reads) within the flanked repeat. 
We have, 
\begin{eqnarray}
{\cal P}({\cal E}) \le \epsilon_{1}+\epsilon_{2}+\delta_{cov}\label{eq:performance-1}
\end{eqnarray}
Therefore, to guarantee confidence in the phasing procedure,
 it suffices to upper bound $\epsilon_{1}$, $\epsilon_{2}$ and $\delta_{cov}$.  
We tabulate the error probabilities of $\epsilon_{1}$, $\epsilon_{2}$  in Table \ref{tab:Calibration-of-error} 
for phasing a flanked repeat (whose length is 5000 whereas the genome length is 5M). 
The flanked repeat has two sites of polymorphism which partition it into three equally spaced segments. 

From Table \ref{tab:Calibration-of-error}, when $p=0.01$, the information requirement
translates to the condition of having three bridging reads spanning the shorter exact repeat of 
the longest pair of exact interleaved repeats. 
Therefore, the information requirement for  X-phased Multibridging shown in Fig \ref{noisyPlot} 
also corresponds to this condition.  
It is noted that  X-phased Multibridging has the same vertical asymptote as the 
noiseless lower bound. 
The vertical shift is due to the increase of requirement on the number of bridging reads from $k=1$ (noiseless case) to $k=3$ (noisy case).
\begin{table}
\begin{centering}
\subfloat[\label{e1}Calibration for $\epsilon_{1}$]{%
\begin{tabular}{|c|c|c|}
\hline 
$p$    & Coverage (NL/G)  & $\epsilon_{1}$\tabularnewline
\hline 
\hline 
0.01    & 20  & 0.00\tabularnewline
\hline 
0.01   & 40  & 0.00\tabularnewline
\hline 
0.01    & 60  & 0.00\tabularnewline
\hline 
0.1    & 20  & 0.16\tabularnewline
\hline 
0.1    & 40  & 0.00\tabularnewline
\hline 
0.1    & 60  & 0.00\tabularnewline
\hline 
\end{tabular}

}

\subfloat[\label{e2}Calibration for $\epsilon_{2}$]{%
\begin{tabular}{|c|c|c|}
\hline 
$p$  & Number of bridging reads $k$  & Upper bound for $\epsilon_{2}$\tabularnewline
\hline 
\hline 
0.01  & 1  & 0.060\tabularnewline
\hline 
0.01  & 3  & 0.0036\tabularnewline
\hline 
0.01  & 5  & 0.00024\tabularnewline
\hline 
0.1  & 11  & 0.089\tabularnewline
\hline 
0.1  & 21  & 0.022\tabularnewline
\hline 
0.1  & 31  & 0.0059\tabularnewline
\hline 
\end{tabular}

}
\par\end{centering}

\caption{\label{tab:Calibration-of-error}Calibration of error probability
made by the phasing procedure of X-phased Multibridging}
\end{table}

\section*{Simulation of the prototype assembler}
Based on the algorithmic design presented, we implement a prototype assembler for automatic genome finishing using reads corrupted by substitution noise. First, the assembler was tested on synthetic genomes, which were generated according to the genome model described previously.  This demonstrates a proof-of-concept that one can achieve genome finishing with read length close to $\ell_{\rm crit}$, as shown in Fig \ref{fig:Simulation-of-the}.  The number on the line represents the number of simulation rounds (out of 100) in which the reconstructed genome is a single contig with  $\ge 99 \%$ of its content matching the ground truth.

Second, the assembler was tested using synthetic reads, sampled from genome ground truth downloaded from NCBI.   The assembly results  are
shown in Table \ref{tab:Assembly-of-several}.
The observation from the simulation result is that we can assemble genomes to 
finishing quality with information requirement near the noiseless lower bound.
More information about the detail design of the prototype assembler is presented in the Appendix 
and source code/data set can be found in \cite{kkcode}.

\begin{table*}[!tbh]
\begin{centering}
\begin{tabular}{|c|c|c|c|c|c|c|c|c|c|c|c|c|}
\hline 
Index  & Species  & $G$  & $p$  & $\frac{NL}{G}$  & $L$  & $\tilde{\ell}_{\rm max}$  & $\tilde{\ell}_{\rm crit}$ & $\ell_{\rm crit}$  & \% match  & Ncontig  & $\frac{N}{N_{noiseless}}$  & $\frac{L}{\ell_{\rm crit}}$\tabularnewline
\hline 
\hline 
1  & a  & 1440371  & 1.5\%  & 37.36 X  & 930  & 1817  & 803  & 770  & 100.00  & 1  & 1.57  & 1.21\tabularnewline
\hline 
2  & a  & 1440371  & 1.5\%  & 33.14 X  & 970  & 1817  & 803  & 770  & 99.95  & 1  & 1.67  & 1.26\tabularnewline
\hline 
3  & a  & 1440371  & 1.5\%  & 29.60 X  & 1000  & 1817  & 803  & 770  & 99.99  & 1  & 1.66  & 1.30\tabularnewline
\hline 
4  & b  & 1589953  & 1.5\%  & 40.82 X  & 2440  & 4183  & 2155  & 2122  & 100.00  & 1  & 1.30  & 1.15\tabularnewline
\hline 
5  & b  & 1589953  & 1.5\%  & 21.31 X  & 2752  & 4183  & 2155  & 2122  & 99.99  & 1  & 1.19  & 1.30\tabularnewline
\hline 
6  & b  & 1589953  & 1.5\%  & 20.66 X  & 2900  & 4183  & 2155  & 2122  & 99.99 & 1  & 1.35  & 1.37\tabularnewline
\hline 
7  & c  & 1772693  & 1.5\%  & 30.03 X  & 3950  & 5018  & 3234  & 3218  & 99.96  & 1  & 1.36  & 1.23\tabularnewline
\hline 
8  & c  & 1772693  & 1.5\%  & 21.96 X  & 4279  & 5018  & 3234  & 3218  & 99.97  & 1  & 1.33  & 1.33\tabularnewline
\hline 
9  & c  & 1772693  & 1.5\%  & 17.03 X  & 4700  & 5018  & 3234  & 3218  & 100.00  & 1  & 1.31  & 1.46\tabularnewline
\hline 
10  & d  & 1006701  & 1.5\%  & 35.23 X  & 6867  & 15836  & 10518  & 5494  & 99.05  & 1  & 1.72  & 1.25\tabularnewline
\hline 
11  & d  & 1006701  & 1.5\%  & 19.88 X  & 7500  & 15836  & 10518  & 5494  & 97.86  & 1  & 1.30  & 1.37\tabularnewline
\hline 
12  & d  & 1006701  & 1.5\%  & 17.69 X  & 9000  & 15836  & 10518  & 5494  & 98.10  & 1  & 1.68  & 1.64\tabularnewline
\hline 
\end{tabular}
\par\end{centering}

\caption{\label{tab:Assembly-of-several} Simulation results on the assembly of several real genomes using reads corrupted by substitution noise ((a)
\textit{Prochlorococcus marinus} (b) \textit{Helicobacter pylori} (c) \textit{Methanococcus maripaludis}
(d) \textit{Mycoplasma agalactiae}) with $\ell_{\rm crit}=\max(\ell_{\rm int},\ell_{\rm tri})$
, $\tilde{\ell}_{\rm crit}=\max(\tilde{\ell}_{\rm int},\tilde{\ell}_{\rm tri})$
and $N_{noiseless}$ is the lower bound on number of reads in the
noiseless case for $1-\epsilon=95\%$ confidence recovery}
\end{table*}

\begin{figure}
\includegraphics[width=7cm]{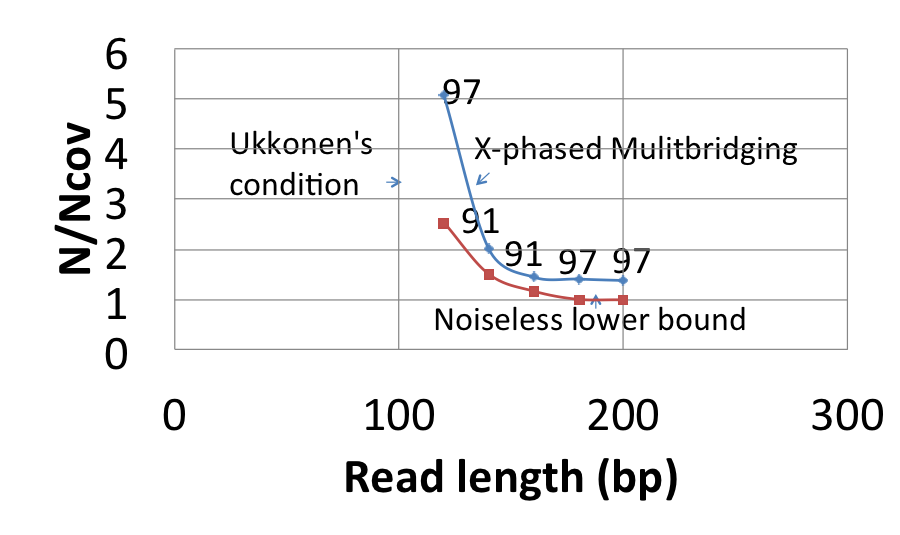} 

\begin{minipage}[t]{7.5cm}
\caption{Simulation results on the assembly of synthetic genomes using reads corrupted by substitution noise\label{fig:Simulation-of-the}. The parameters are as follows. $G=10K; \tilde{\ell}_{\rm max} =\ell_{\rm max} = 500, \tilde{\ell}_{\rm int} = 200,  \ell_{\rm int}=100$ with two sites of polymorphism within the flanked repeat. $p=1.5\%,\epsilon =5 \%  $. } 
\end{minipage}
\end{figure}

\begin{figure}[!htb]
\subfloat[Form K-mer Clusters\label{fig:Form-K-mer-Clusters}]{\includegraphics[width=5cm]{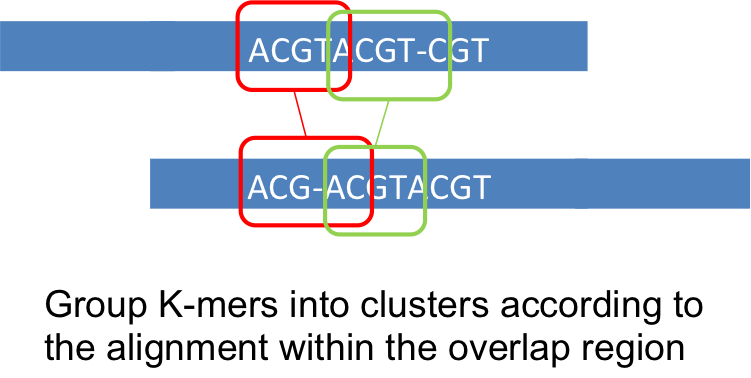}}

\subfloat[Abnormality in (indel) De Bruijn Graph\label{fig:Abnormalty-in-Noisy(indel)} ]{\includegraphics[width=7cm]{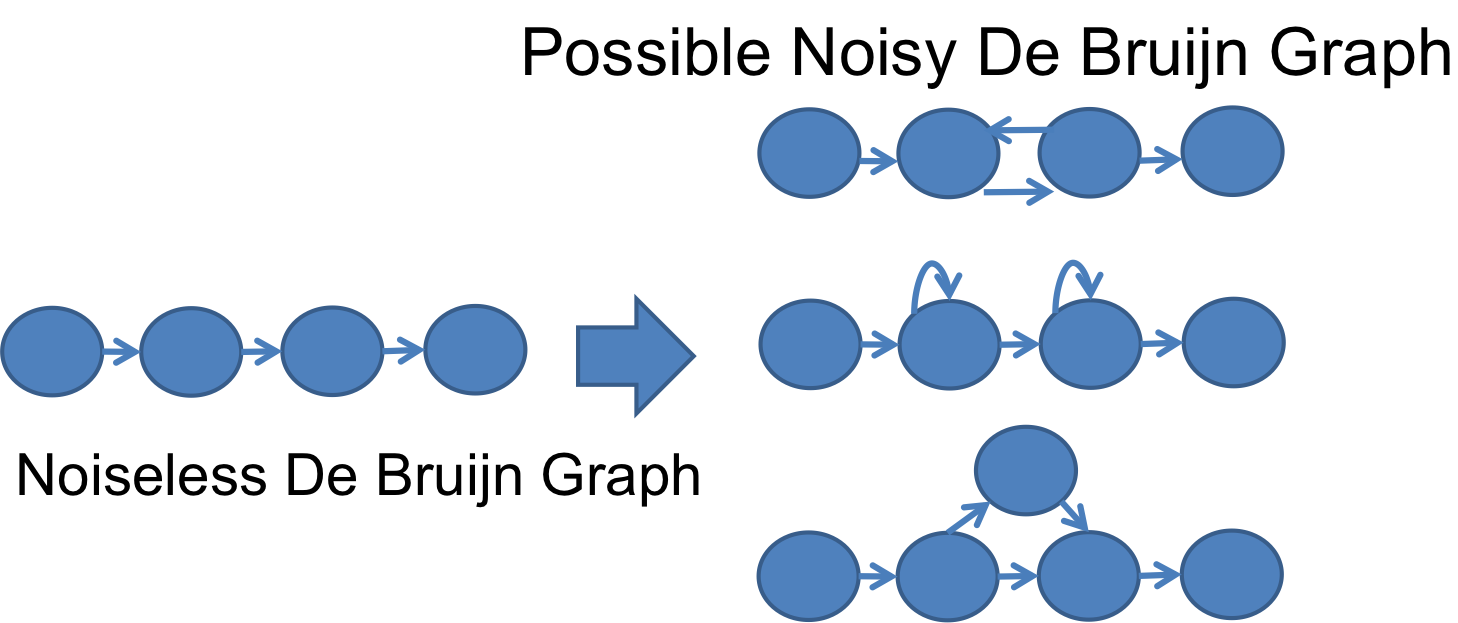}}

\begin{minipage}{7cm}
\caption{Treatment of reads corrupted by indel noise}  
\end{minipage}

\end{figure}
\section*{Extension to handle indel noise}

A further extension of the prototype assembler addresses the case of reads corrupted by indel noise.
Similar to the case of substitution noise, tests were performed on synthetic reads sampled from real genomes and synthetic genomes. 
Simulation results are summarized in Table \ref{tab:Simulation-Results-for}
where $p_{i},p_{d}$ are insertion probability and deletion probability and
rate is the number of successful reconstruction(i.e. simulation rounds that show mismatch $<5\%$) divided by total number of simulation rounds.
The simulation result for indel noise corrupted reads shows that X-phased Multibridging can be generalized to assemble indel noise corrupted reads.  The information requirement for automated finishing is about a factor of two from the noiseless lower bound for both N and L.

We remark that one 
non-trivial generalization is the way that we form the noisy De Bruijn graph for K-mer clusters. 
In particular, we first compute the pairwise overlap alignment among reads, then we
use the overlap alignment to group K-mers into clusters. Subsequently,
we link successive cluster of K-mers together as we do in Alg \ref{alg:Noisy-Multi-Bridging-1}.
An illustration is shown in Fig \ref{fig:Form-K-mer-Clusters}. However,
due to the noise being indel in nature, the edges in the noisy De
Bruijn graph may point in the wrong direction as shown in Fig \ref{fig:Abnormalty-in-Noisy(indel)}.
In order to handle this, we traverse the graph and remove such abnormality when they are detected. 

\begin{table*}
\begin{centering}
\begin{tabular}{|c|c|c|c|c|c|c|c|c|c|c|c|}
\hline 
Type & $G$  & $p_{i}$  & $p_{d}$  & $\frac{NL}{G}$  & $L$  & $\tilde{\ell}_{\rm max}$  & $\tilde{\ell}_{\rm crit}$  & $\ell_{\rm crit}$  & $\frac{N}{N_{noiseless}}$  & $\frac{L}{\ell_{\rm crit}}$ & Rate\tabularnewline
\hline 
\hline 
Synthetic & 50000 & 1.5\%  & 1.5\%  & 23.0 X & 200 & 500 & 200 & 100 & 2.25 & 2  & 28/30\tabularnewline
\hline 
Synthetic & 50000 & 1.5\%  & 1.5\%  & 24.1 X & 180 & 500 & 200 & 100 & 2.33 & 1.8 & 27/30\tabularnewline
\hline 
a & 1440371 & 1.5\%  & 1.5\%  & 28.53 X & 1000 & 1817  & 803  & 770  & 1.60 & 1.30& 1/1\tabularnewline
\hline 
b & 1589953 & 1.5\%  & 1.5\%  & 20.66 X & 2900 & 4183  & 2155  & 2122  & 1.35  & 1.37 & 1/1\tabularnewline
\hline
\end{tabular}
\par\end{centering}

\caption{\label{tab:Simulation-Results-for}Simulation results on the assembly of real/synthetic genomes using reads corrupted by indel noise(Synthetic: randomly generated to fit  $\tilde{\ell}_{\rm max}, \tilde{\ell}_{\rm crit}, \ell_{\rm crit}$; (a) : \textit{Prochlorococcus marinus} ; (b): \textit{Helicobacter pylori})}
\end{table*}

\section*{Conclusion}
In this work, we show that even when there is noise in the reads, 
one can successfully reconstruct with information requirements close to the noiseless 
fundamental limit. A new assembly algorithm, X-phased Multibridging,  
is designed based on a probabilistic model of the genome. It is shown through analysis  
to perform well on the model, and through simulations to perform well  on real genomes.

The main conclusion of this work is that, with an appropriately designed assembly algorithm, the information requirement for genome assembly is insensitive to moderate read noise. 
We believe that the information theoretic insight is useful to guide the design of future assemblers. 
We hope that these insights allow future assemblers to better leverage the high throughput sequencing read data to provide higher quality assembly.

\newpage{}  \bibliographystyle{plain}

\begin{backmatter}

\section*{Competing interests}
The authors  K.K.L and A.K are or were employees of Pacific Biosciences, a company commercializing DNA sequencing technologies at the time that this work was completed.

\section*{Author's contributions}
  K.K.L, A.K and D.T  performed research and wrote the manuscript. K.K.L implemented the algorithms and performed the experiments.

\section*{Acknowledgements}
  The assembly experiments were partly done on the computing infrastructure of Pacific Biosciences. 
  
\section*{Declaration}
  The authors K.K.L and D.T. are partially supported by the Center for Science of Information (CSoI), an NSF Science and Technology Center, under grant agreement CCF-0939370.

\section*{Additional Files}
  \subsection*{Additional file 1}
 Details of the proofs, in-depth description of the design of the prototype assembler and details of simulation results are presented.

\end{backmatter}

\bibliographystyle{plain}
\bibliography{myrefs}

\newpage{}
\appendix

\section*{Appendix : Proof on Performance Guarantee}

Here we use the short hand of $l_{repeat}$, $l_{interleaved}$ and
$l_{triple}$ to represent the corresponding approximate repeat length
of longest simple, interleaved, triple repeat respectively.

\subsection*{Greedy Algorithm}

Let us define a $\theta-$ neighborhood of the repeat specified by
$x[a:b]$ and $x[c:d]$ to be the loci of repeat which
are $x[a-\theta:b+\theta]$ and $x[c-\theta:d+\theta]$.  

We say a repeat is $\theta-$bridged if there exists a read that cover
the $\theta$-neighborhood of at least one copy of the repeat. For
simplicity of arguments, we assume $l_{repeat}>>\max(l_{intereleave},l_{triple})$.
\begin{lemma}
\label{noisyGreedyCondition}We first note the following sufficient
conditions for Noisy Greedy to succeed.
\begin{enumerate}
\item Merging at stages from $L$ to $\ell_{\rm iid}(p,\epsilon,G)$ are merging
successive reads 
\item Every successive reads have overlap with length at least $\ell_{\rm iid}(p,\frac{\epsilon}{3},G)$
\end{enumerate}
\end{lemma}

\begin{theorem}
Under the generative model on genome, with $\ell_{\rm iid}=\ell_{\rm iid}(p,\frac{\epsilon}{3},G)$,
$\alpha=\alpha(p,\frac{\epsilon}{3},G)$, if 
\begin{align*}
L & >l_{repeat}+2\cdot \ell_{\rm iid}\\
G & >N>\max(\frac{G\ln\frac{3}{\epsilon}}{L-l_{repeat}-2\cdot \ell_{\rm iid}},\frac{G\cdot\ln\frac{N}{\epsilon/3}}{L-\ell_{\rm iid}})
\end{align*}
, then ${\cal P}({\cal S}^{C})\le\epsilon$.\end{theorem}

\begin{proof}
In order to prove that claim, let us break down into several subparts

Let $E_{1}$ be the event that condition 1 in Lemma (\ref{noisyGreedyCondition})
is not satisfied. $E_{2}$ be the event that condition 2 in Lemma
(\ref{noisyGreedyCondition}) is not satisfied. $E_{3}$ be the event
that the long/interleave/triple repeat is not $\ell_{\rm iid}-$bridged. 

Now we claim that with the chose $(N,L)$ in the range, 

1. ${\cal P}(E_{1})\le\frac{\epsilon}{3}$

2. ${\cal P}(E_{3})\le\frac{\epsilon}{3}$

3. ${\cal P}(E_{2}\mid E_{1}^{C}\cap E_{3}^{C})\le\frac{\epsilon}{3}$

We first see how we can use these to obtain the desired claim and
proceed to prove each of the above sub-claims.

\begin{align*}
{\cal P}(S^{C})= & {\cal P}(E_{1})+{\cal P}(E_{2}\cap E_{1}^{C})\\
= & {\cal P}(E_{1})+{\cal P}(E_{2}\cap E_{1}^{C}\cap E_{3}^{C})+{\cal P}(E_{2}\cap E_{1}^{C}\cap E_{3})\\
\le & {\cal P}(E_{1})+{\cal P}(E_{2}\mid E_{1}^{C}\cap E_{3}^{C})+{\cal P}(E_{3})\\
\le & \frac{\epsilon}{3}+\frac{\epsilon}{3}+\frac{\epsilon}{3}\\
= & \epsilon
\end{align*}

Now, we proceed to prove each of the sub-claims. 

1. With $N>\frac{G\cdot\ln\frac{N}{\epsilon/3}}{L-\ell_{\rm iid}}$, we have,

\begin{align*}
{\cal P}(E_{1})\le & N\exp(-\frac{N}{G}(L-\ell_{\rm iid}))\\
\le & \frac{\epsilon}{3}
\end{align*}

2. With $N>\frac{G\cdot\ln\frac{3}{\epsilon}}{L-l_{repeat}-2\cdot \ell_{\rm iid}}$we
have, 

\begin{align*}
{\cal P}(E_{3})\le & \exp(-\frac{N}{G}(L-2\ell_{\rm iid}-l_{repeat}))\\
\le & \frac{\epsilon}{3}
\end{align*}

3. With the choice of $\ell_{\rm iid}=\ell_{\rm iid}(p,\frac{\epsilon}{3},G)$,
we have,

\begin{align*}
{\cal P}(E_{2}\mid E_{1}^{C}\cap E_{3}^{C})\le & N^{2}\cdot\exp(-\ell_{\rm iid}D(\alpha||\frac{3}{4}))\\
 & +2N\exp(-\ell_{\rm iid}D(\alpha||\eta))\\
\le & G^{2}\cdot\exp(-\ell_{\rm iid}D(\alpha||\frac{3}{4}))\\
 & +2G\exp(-\ell_{\rm iid}D(\alpha||\eta))\\
\le & \frac{\epsilon}{3}
\end{align*}

Here we use the fact that there are indeed 4 types of overlap as in
Fig \ref{fig:Overlap-Type}. And given the bridging condition, we
are only left with 2 types, namely, both ending segments outside/exactly
one ending segment outside the longest repeat repeat region. 

\begin{figure}
\begin{centering}
\includegraphics[width=8cm]{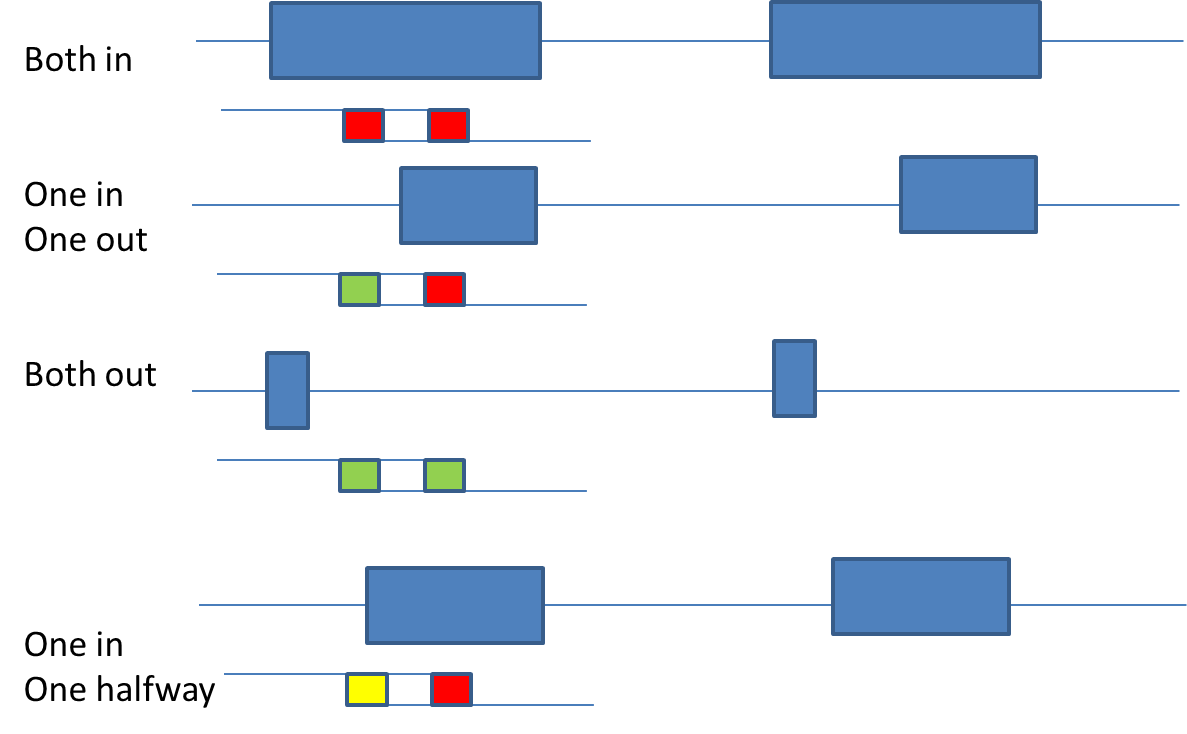}
\par\end{centering}

\caption{\label{fig:Overlap-Type}Overlap Type}
\end{figure}

\end{proof}

\subsection*{Simple De Bruijn Algorithm}

Before continuing proving the performance of Multibridging Algorithm,
it is instructive to analyze the following Simple De Bruijn Algorithm
(Alg \ref{alg:Noisy-Simple-De Bruijn}) because this is closely related
to the Multibridging Algorithm. 

\begin{algorithm}[tbh]
\raggedright{}
0. Choose K to be $\max(l_{interleave},l_{triple})$

1. Extract Kmers from reads

2. Clusters Kmers

3. Form Kmer Graphs

4. Condense the graph

5. Clear Branches

6. Condense graph

7. Find Euler Cycle

\caption{\label{alg:Noisy-Simple-De Bruijn}Noisy Simple De Bruijn}
\end{algorithm}

Here we first define several genomic region of interest which we will
refer to in the proofs below. 

a) $S_{0}$ = set of K-mers that are completely inside $\ell_{\rm iid}-$
neighborhood of the longest repeat 

b) $S_{1}$ = set of K-mers that are completely inside the longest
repeat

c) $S_{2}=S_{0}\backslash S_{1}$
\begin{lemma}
\label{simpleDe BruijnCondition}Here we provide several deterministic
conditions that guarantee the success of the algorithm.

1. Successive reads overlap with length at least $K$

2 K-mers are almost correctly clustered, that is, 

\textup{a) K-mers from the same locus but not merged }

b)\textup{ $x$ not in $S_{0}$ s.t. $x$ get clustered with wrong
K-mers}

c)\textup{ $x$ in $S_{0}$ s.t.$x$ get clustered with elements other
than its own cluster/mirror cluster(mirror cluster is defined to be
the cluster for the other copy of the repeat)}

3) Repeat at both circle are at least $2\cdot \ell_{\rm iid}$ separated(the
interleaving segments between the repeat differ with at least $2\ell_{\rm iid}$in
length)\end{lemma}
\begin{proof}
We note that every length K segments $x\not\in S_{0}$, they are represented
as a distinct node in the K-mer graph because of the length K that
we pick and the condition that successive reads overlap at least K
bases. Moreover, for K-mers $x\in S_{1}$, they are condensed into
the repeat as 'X' in Fig \ref{fig:Before-branch-clearing}. However,
for the K-mers $x\in S_{2}$, they have chances not to merge properly,
thus they form into the branches surrounding 'X' in Fig \ref{fig:Before-branch-clearing}.
Because of condition 3, branch clearing will not eliminate the 'A'
or 'C' in Fig \ref{fig:Before-branch-clearing}, further after condensing,
we get the desired K-mer graph as in Fig \ref{fig:After-branch-clearing}
and this can be successfully read by a Eulerian Walk.
\end{proof}
\begin{figure}
\begin{centering}
\subfloat[\label{fig:Before-branch-clearing}Before branch clearing]{\begin{centering}
\includegraphics[width=8cm]{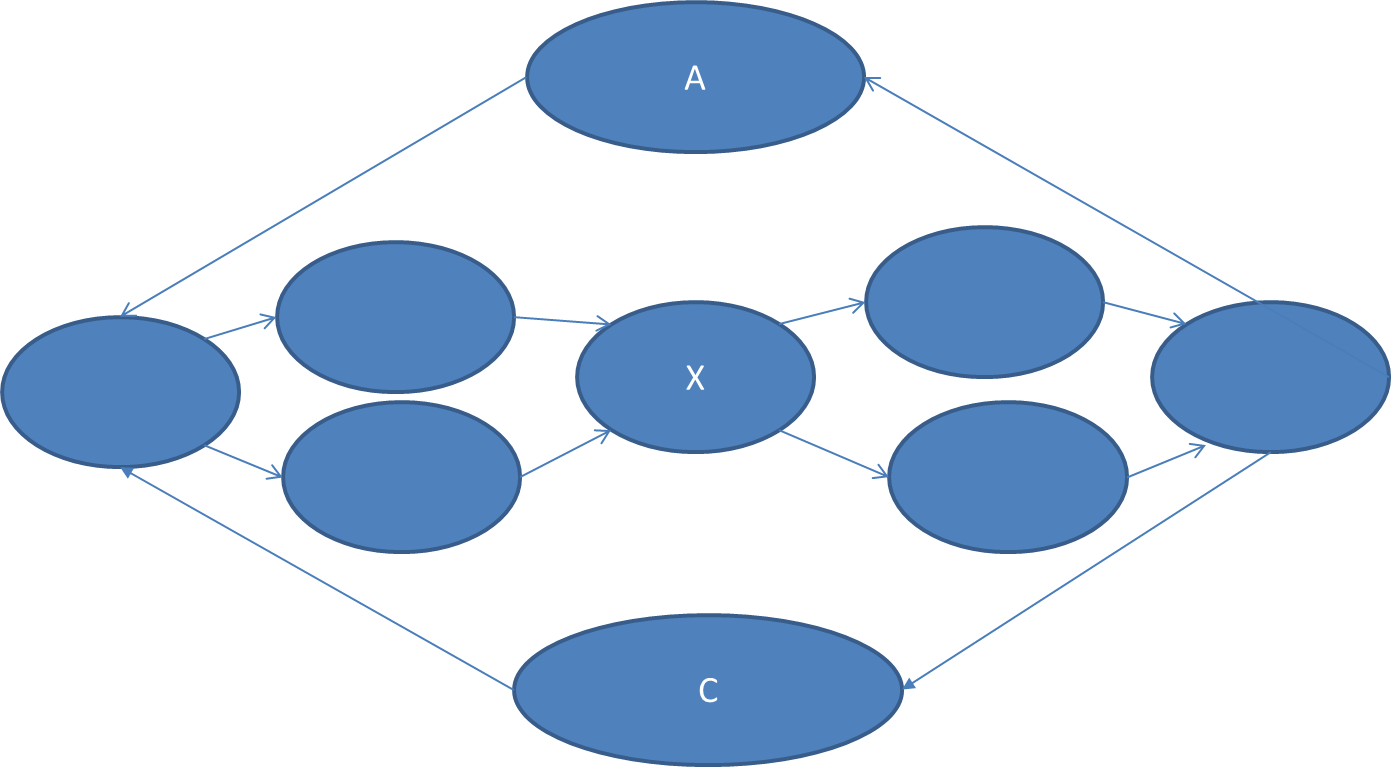}
\par\end{centering}

}
\par\end{centering}

\begin{centering}
\subfloat[\label{fig:After-branch-clearing}After branch clearing]{\begin{centering}
\includegraphics[width=1cm]{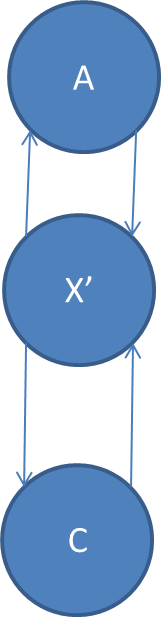}
\par\end{centering}

\centering{}}
\par\end{centering}

\caption{Branch clearing}
\end{figure}

\begin{theorem}
\label{simpleDe BruijnThm}If $G>\frac{6}{\epsilon \ell_{\rm iid}}$$,G\ge N\ge\frac{G\cdot\ln\frac{3N}{\epsilon}}{L-\max(l_{int},l_{triple})-2\ell_{\rm iid}}$
, with $\ell_{\rm iid}=\ell_{\rm iid}(p,\frac{\epsilon}{3},G)$ $\alpha=\alpha(p,\frac{\epsilon}{3},G)$,
then, ${\cal P}(S^{C})\le\epsilon$\end{theorem}
\begin{proof}
We first note that in order to obtain a bound on the error probability,
we only need to separately bound the probability that each of the
conditions in Lemma \ref{simpleDe BruijnCondition} fail,which are
$\le\frac{\epsilon}{3}$ each. Thus, combining, we get, ${\cal P}(S^{C})\le\epsilon$.
\end{proof}

\subsection*{Multibridging Algorithm}

An illustration of noisy Multibridging Algorithm is shown in Fig (\ref{illustrationMBA}).

\begin{figure}
\begin{centering}
\includegraphics[width=8cm]{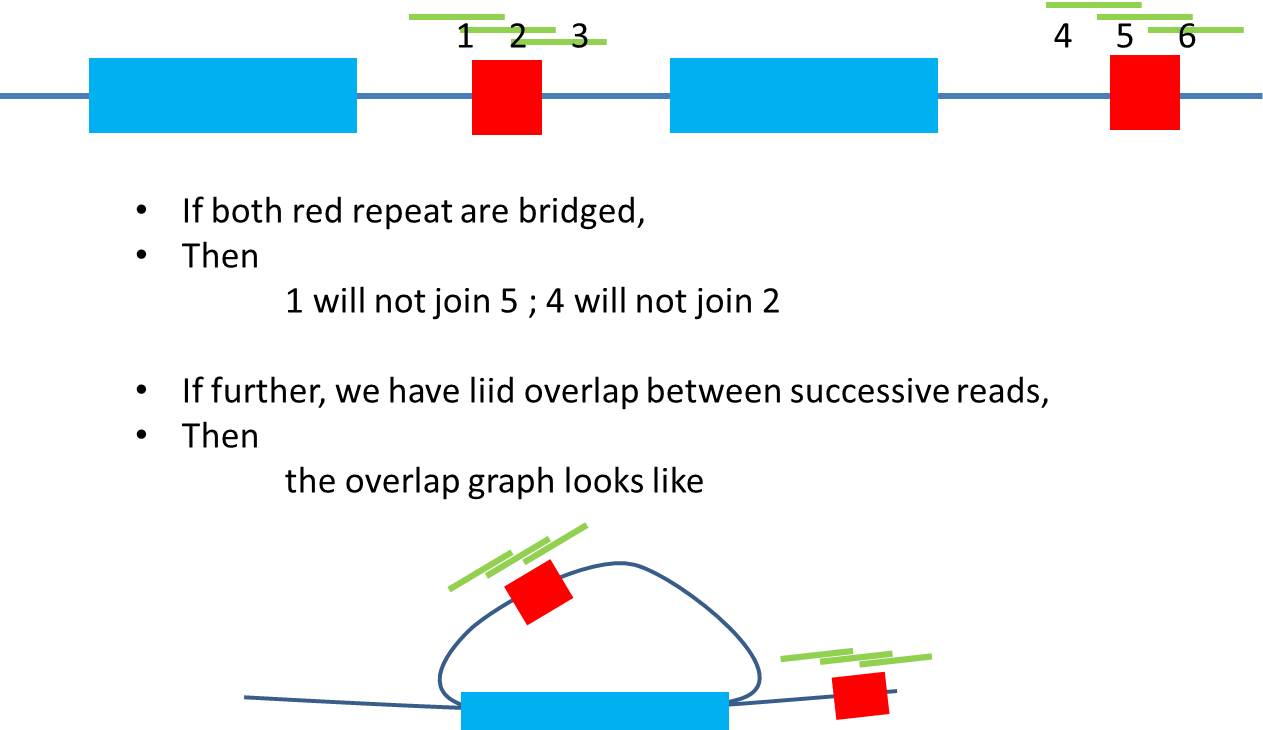}
\par\end{centering}

\caption{\label{illustrationMBA}An illustration of the Noisy Multi Bridging }
\end{figure}

\begin{lemma}
\label{MBAConditionDetail}Here are the deterministic conditions for
the algorithm to succeed.

1) Every successive reads overlap at least $\ell_{\rm iid}(p,\frac{\epsilon}{3},G)$

2) K-mers are almost correctly clustered, that is,

\textup{a) K-mers from the same locus but not merged }

b)\textup{ $x$ not in $S_{0}$ s.t. $x$ get clustered with wrong
K-mers}

c)\textup{ $x$ in $S_{0}$ s.t. $x$ get clustered with elements
other than its own cluster/mirror cluster}

3) Repeat at both circle are at least $2\cdot \ell_{\rm iid}$ separated

4) When finding successors/predecessors, they are the real successors
and predecessors\end{lemma}
\begin{proof}
Along the same lines as the proof in Lemma (\ref{simpleDe BruijnCondition}),
we only note that in this algorithm, we have an extra step of finding
predecessor/successors. Moreover, the overlap here is significantly
reduced to only $\ell_{\rm iid}$ instead of K in the Noisy Simple De Bruijn
case.\end{proof}
\begin{theorem}
With $G>\frac{6}{\epsilon \ell_{\rm iid}}$, $G\ge N\ge\max(\frac{G}{L-2\ell_{\rm iid}}\ln\frac{N}{\epsilon/3},\frac{G\ln\frac{3}{\epsilon}}{L-\max(l_{triple},l_{interleave})-2\ell_{\rm iid}})$
,with $\ell_{\rm iid}=\ell_{\rm iid}(p,\frac{\epsilon}{3},G)$ $\alpha=\alpha(p,\frac{\epsilon}{3},G)$
, then ${\cal P}(S^{C})\le\epsilon$.\end{theorem}
\begin{proof}
Here we note that with the given coverage, bridging conditions of
the interleave repeat and the triple repeat are satisfied. And when
this is true, then Condition 4 in Lemma \ref{MBAConditionDetail}
is true with high probability. Following the arguments in Theorem
\ref{simpleDe BruijnThm}, we get desired.
\end{proof}

\section*{Appendix: Design and additional algorithmic components for the prototype
assembler}

\subsection*{Pipeline of the prototype assembler}

The pipeline of the prototype assembler is shown in Fig \ref{fig:Pipeline-of-the}.
With a ground truth genome as input, the output is the performance
of the whole pipeline by giving the mismatch rate.

\begin{figure}
\includegraphics[width=12cm]{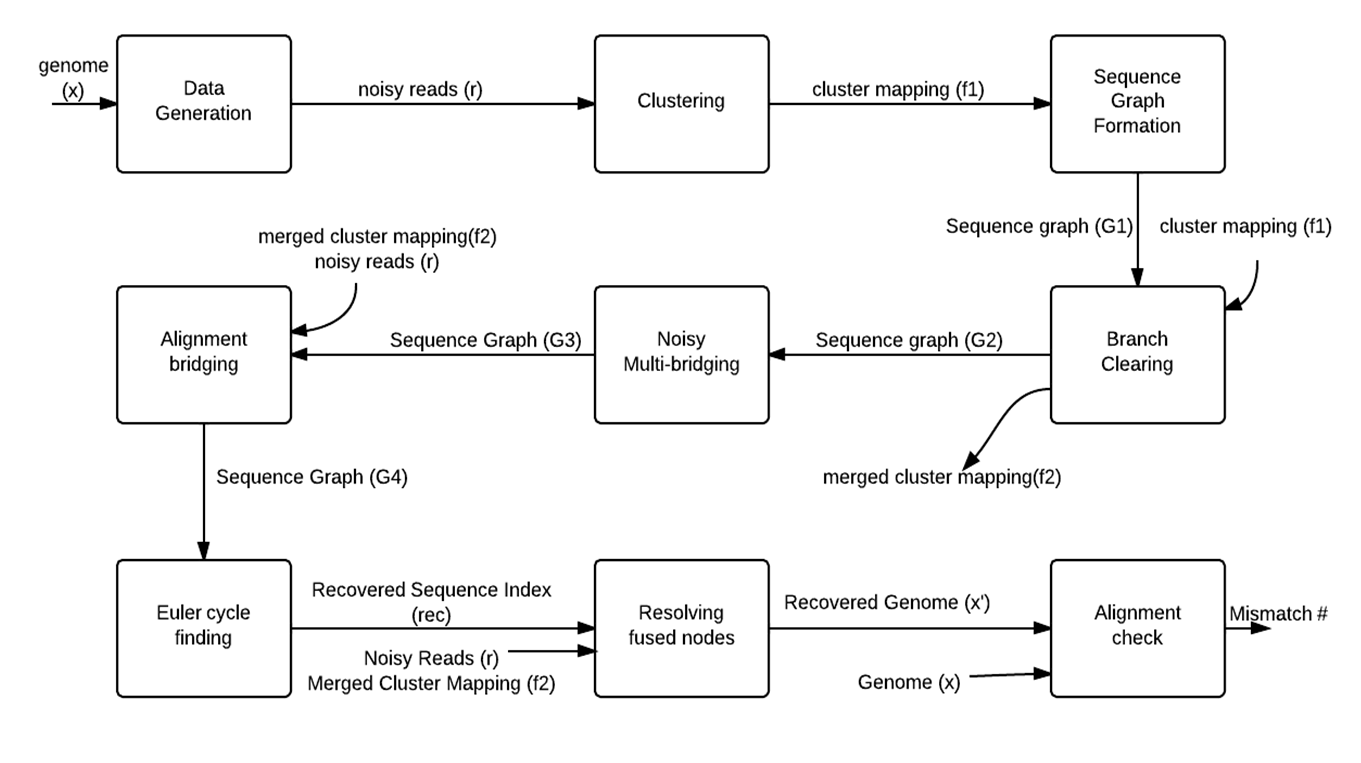}

\caption{\label{fig:Pipeline-of-the}Pipeline of the prototype assembler}

\end{figure}

\subsection*{A more robust branch clearing step}

Since we employ a speed up step in the clustering and there may be
K-mers that are not completely clustered correctly in the clustering
step of Multibridging Algorithm. Regarding that, we need to have a
more robust branch clearing step. In particular, we first classfy
nodes as ``big'' or ``small'' nodes based on the size of the nodes
in the sequence graph. The key idea is to merge the small nodes together
while keeping the big nodes unchanged. Starting from each big nodes,
we tranverse the graph to detect all the small nodes that link the
current big node to other big nodes. Then, we classify the small nodes
into levels(depending on its distance from the current big node).
After that, the small nodes in the same level are merged. Finally,
we note that we keep the reachability among each big nodes.

\subsection*{Enhanced Multibridging Algorithm that can resolve middle range repeats}

We note that the ideas presented here can also be found in the prior
work on the treatment of noiseles sreads. It is stated here for completeness.
In the noisy setting, instead of considering the alphabet set to be
$\Sigma=\{A,C,G,T\}$, one can consider the alphabet set as the cluster
index of the K-mers. 

\begin{algorithm}
\raggedright{}
Resolution of repeats:

0. Intially the weight of the edge are set to be 1. 

1. While there is a X-node $v$ :

a) For each edge ($p_{i}$,$v$) with weight $a_{p_{i},v}$, create
a new node $u_{i}=^{p_{i}\to}v$ and an edge ($p_{i},u_{i}$) with
weight 1 + $a_{p_{i}v}$. Similarly, for each edge ($v$,$q_{j}$),
create a new node $w_{j}=v^{\to q_{j}}$ and an edge ($w_{j}$, $q_{j}$)

b) If $v$ has a self-loop ($v$, $v$) with weight $a_{v,v}$, add
an edge $(v^{\to v},^{v\to}v)$ with weight $a_{v,v}+2$

c) Remove node $v$ and all incident edges

d) For each pair of $u_{i}$, $w_{j}$ adjacent in a read(extending
to at least length of $\ell_{\rm iid}$ on both sides of the X-node), add
edge ($u_{i},w_{j}$). If exactly on each of the $u_{i}$and $w_{j}$nodes
have no added edge, add the edge. 

e) Condense the graph

\caption{Enhanced Multibridging Algorithm}

\end{algorithm}

\section*{Appendix: Treatment of indel noise}

\subsection*{Formation of K-mer De Bruijn graph for indel corrupted reads}

In order to form K-mer De Bruijn graph for indel corrupted reads, we
first need to have a clear notion of K-mers. We define K-mers to be
the length K segments in the genome ground truth (as opposed to the
usual definition from the reads). Although we mostly work on the reads
themselves, the definition of the Kmers are based on the ground truth.
In order to successfully cluster K-mers, we need to do the following
steps. 

1. We first compute the pairwise alignement of the reads. 

2. Based on the pairwise alignment, for each length K-segments, we
know which should be aligned to which. We then group them together
using the alignment result. 

3. Finally, we end up with the length K segments from the reads clustering
together, and now we use it as an operational way to identify the
Kmers since each cluster will naturally correspond to a K-mers originated
from the genome groundtruth(though there are a few discrepancy, mostly
this is correct).

4. After we identify the K-mers clusters, we add an edge between them
if there exists a read such that there are two consecutive Kmers originate
from it.

\subsubsection*{Graph surgery to clear abnormality of the noisy De Bruijn graph}

Due to indel noise and runs of the same alphabet, the way that we
form K-mers graph may need to abnormality of the graph. We thus perform
a graph tranversal and identify the abnormality that are of short
length(i.e. resulted from noise but not the genome structure). After
that, we remove such abnormality. This step also involves transitive
edge reduction and removal of small self loops.

\subsection*{X-phased step tailored for indel noise type }

\subsubsection*{Generalization to handle Indel Error }

When dealing with indel noise, the neighborhood of reads can also
affect consensus of the base. There we have to do sequence alignment
in order to find the appropriate posterior probability in order to
do a maximum likelihood estimate of whether a particular given genomic
location is a site of polymorphism or not. In order to do that, we formulate the problem
as a ML problem as follows. 

\begin{eqnarray}
\max_{T\in\Omega} & \Pi_{i\in S}P(R_{i}\mid T), & P_{err}=P_{opt} \label{1}  \\ 
 \max_{T\in\Omega'} & \Pi_{i\in S}P(R_{i}\mid T), & P_{err}=P_{opt}+\delta_{1} \label{2} \\
\max_{T\in\Omega'} & \Pi_{i\in S'}P(R_{i}\mid T), & P_{err}=P_{opt}+\delta_{1} \label{3}  \\
\max_{T\in\Omega'} & \Pi_{(j,k)}\Pi_{i\in S_{jk}'}P(R_{i}\mid T_{j}^{j+k}), & P_{err}=P_{opt}+\delta_{1}\label{4}  \\
\max_{T\in\Omega'} & \Pi_{(j,j+1)}\Pi_{i\in S_{j,j+1}'}P(R_{i}\mid T_{j}^{j+1}), & P_{err}=P_{opt}+\delta_{1}+\delta_{2} \label{5} 
\end{eqnarray}

Here we also discuss about the places that we take approximation to
enhance the computational efficiency in the steps of the previous
reduction. From (\ref{1}) to (\ref{2}), we use some heuristics to find out the
possible location of SNPs within the whole repeat in which disagreement
is observed after several rounds of error correction. From (\ref{2}) to
(\ref{3}), we remove all the reads that only span one single SNPs and it
has no effect on the error of the detection problem that we are trying
to solve. From (\ref{3}) to (\ref{4}), we further partition the reads into group
in which $S_{jk}'$ is the set of reads that only span the SNPs j
to j+k. Doing this can decompose the ML problem into smaller subproblems
with no effect on the accuracy. Finally, in practice, we take a first
order approximation of (\ref{4}) to (\ref{5}) by only onsidering two SNPs for
each subproblem.

As for each of the marginal probability distribution, the best way
is to run Sum-Product algorithm to compute in a dynamic programming
fashion similar to S-W alignment. But as pointed out in Quiver, this
steps can be significanly speeded up using a Viterbi approximation
and this is also what we implemented in the simulation code.

\subsubsection*{Simulation study}

We simulated on both synthetic and real data set with indel noise
and on a double stranded DNA. In the simulation, we assume that the
reads from the neighborhood of a repeat is given and our goal is to
decide how to extend the reads to span the repeat copies into the
flanking region correctly. The correctness is evaluated based on whether
they can correctly extend the correct reads into the flanking region. 

\begin{table}
\begin{centering}
\begin{tabular}{|c|c|c|c|c|c|c|c|c|c|c|c|}
\hline 
Repeat Type  & $C_{s}$  & $L_{s}$  & $C_{l}$  & $L_{l}$  & $p_{del}$  & $p_{ins}$  & $G$  & Homology  & $l^{approx}$  & $l^{exact}$  & Success \%\tabularnewline
\hline 
\hline 
Randomly generated  & 50X  & 100  & 50X  & 240  & 10\%  & 10\%  & 10000  & 0.67\%  & 300  & 150  & 99\%\tabularnewline
\hline 
A repeat of Ecoli-K12  & --  & --  & 80X  & 3000  & 10\%  & 10\%  & 4646332  & 0.48\%  & 5182  & 1507  & 89\% \tabularnewline
\hline 
A repeat of Bacillus anthracis  & --  & --  & 80X  & 3500  & 10\%  & 10\%  & 5227293  & 0.23\%  & 4778  & 2305  & 85\% \tabularnewline
\hline 
A repeat of Meiothermus ruber  & --  & --  & 80X  & 750  & 10\%  & 10\%  & 3097457  & 1.40\%  & 1217  & 257  & 94\% \tabularnewline
\hline 
\end{tabular}
\par\end{centering}

\caption{\label{tab:contig}Simulation results on long contig creator($C_{s},L_{s}$
are coverage and readlength for short reads. $C_{l},L_{l}$ are coverage
and readlength for long reads. $p_{del},p_{ins}$ are the probability
of insertion and deletion. $G$ is the length of the genome. Homology
is the number of SNPs divided by the length of the approximate repeat.
$l^{approx},l^{exact}$ are the length of the approximate and exact
repeat being studied. Success \% is the percentage of success in 100
rounds) }
\end{table}

\subsection*{Edit distance metric calibration }

We also do a study on whether we can use alignment score to differentiate
between segments from being extracted from the same locus or not.
In Fig \ref{fig:A-calibration-for}, the upper curve is the score
for segment extracted from the same locus while the bottom
curve is for completely iid randomly(irrelevant) generated segment.
And we simulate it for 100 times at each length and the bar indicate
1 standard deviation from the mean. 

\begin{center}
\begin{figure}
\begin{centering}
\includegraphics[width=10cm]{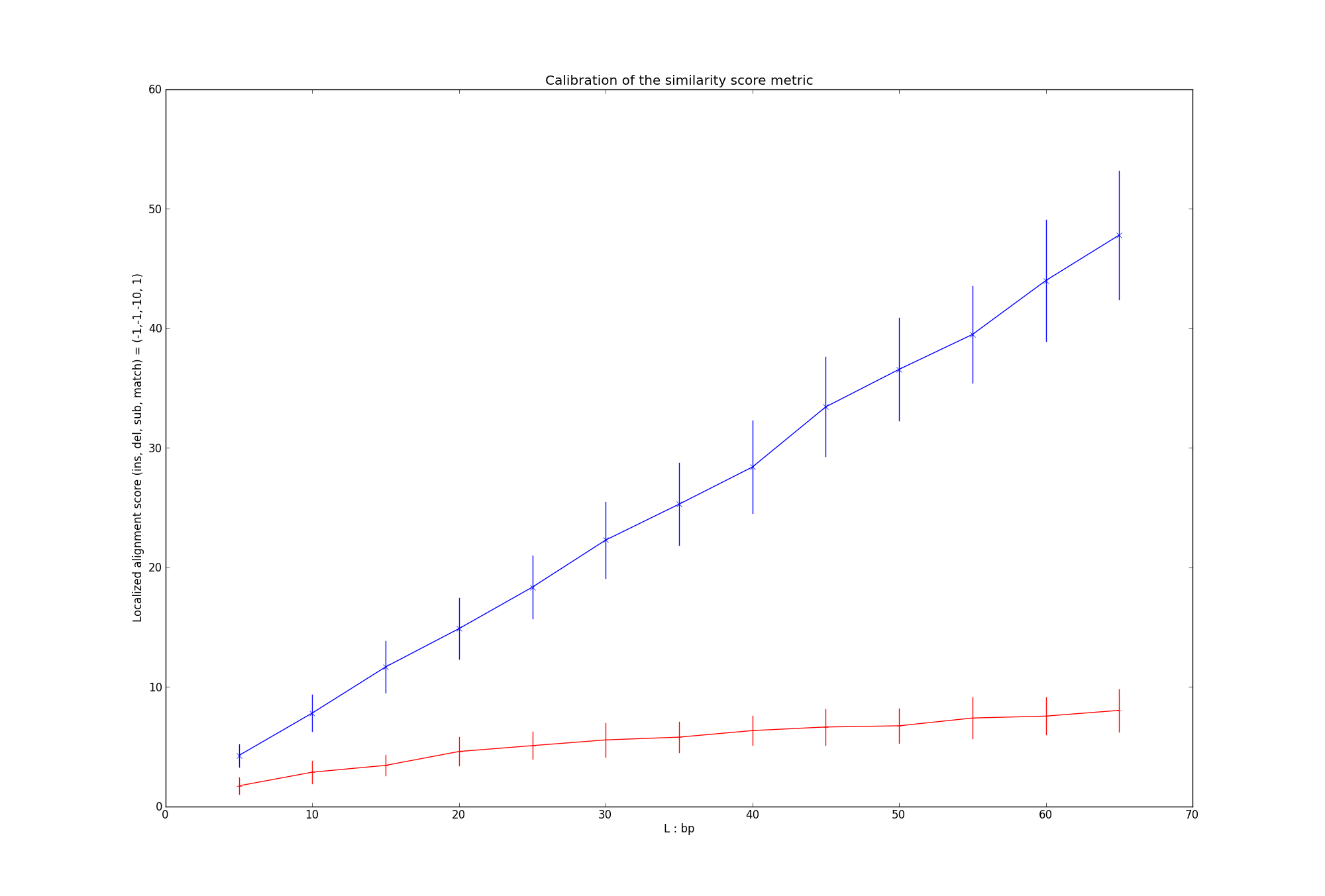}
\par\end{centering}

\caption{\label{fig:A-calibration-for}A calibration for similarity score using
global alignment computation. }

\end{figure}

\par\end{center}

\subsection*{Tolerance in the Multibridging step}
We note that due to the indel noise and the graph surgery that we
perform, an X-node of the graph may be p times longer than the usual
size of the approximate repeat, thus we should have a corresponding
higher tolerance to use the reads to bridge across the repeats. 
\subsection*{Computation speed up of alignment step}
The key bottleneck in computation speed  of the indel extension is on the pairwise alignment of the reads, which can be speeded up using the ideas in BLAST. We use sorting to identify exact matching fingerprint that identify the starting and ending location of the segment that need to be aligned with. After that, we do a local search instead of the whole dynamic programming search. 

\newpage{}

\section*{Appendix : Evidence behind model}

\subsection*{Approximate Repeat}

We let the underlying genome be $\vec{x}$ and use the short hand
that $x[a:b]$ be the $a^{th}$ to $(b-1)^{th}$ entries of $\vec{x}$. 

Let $\vec{v}_{1}=x[s_{1}:s_{1}+l]$ and $\vec{v}_{2}=x[s_{2}:s_{2}+l]$
be two length $l$ substrings of the genome with starting positions
at $s_{1}$ and $s_{2}$ respectively. We call $\vec{v}_{1}$ and
$\vec{v}_{2}$ be an approximate repeat of length $l$ if 

$d(x[s_{1}-W:s_{1}],x[s_{2}-W:s_{2}])\ge0.7W$

$d(x[s_{1}+l:s_{1}+l+W],x[s_{2}:s_{2}+l+W])\ge0.7W$

$d(x[s_{1}-W+k:s_{1}+k],x[s_{2}-W+k:s_{2}+k])<0.7W$ for all $0<k<l$

To understand approximate repeat better, we plot the Hamming distance
for consecutive disjoint window of length 10 as shown in Fig \ref{fig:Example-of-how}
.

\subsection*{Classification of approximate repeat }

While repeats are studied in the literature\cite{bao2002automated},
they are not investigated by looking at the ground truth. This is
partially due to the insufficiency of data in the early days of genome
assembly development. Therefore, based on the ground truth genome,
we define several quantities that allow us to classify approximate
repeat and understand the approximate repeat spectrum of genome. Here
we define stretch and mutation rate. Stretch is defined to be the
ratio of the length ($l^{*}$) of the longest exact repeat within
an approximate repeat divided by the length ($l_{approx}$) of the
approximate repeat. Mutation rate is defined to be number of mutation
within approximate repeat divided by ($l_{approx}-$$l^{*}$). An
illustration is shown in Fig \ref{fig:Example-of-how}.

\begin{figure}
\begin{centering}
\includegraphics[width=8cm]{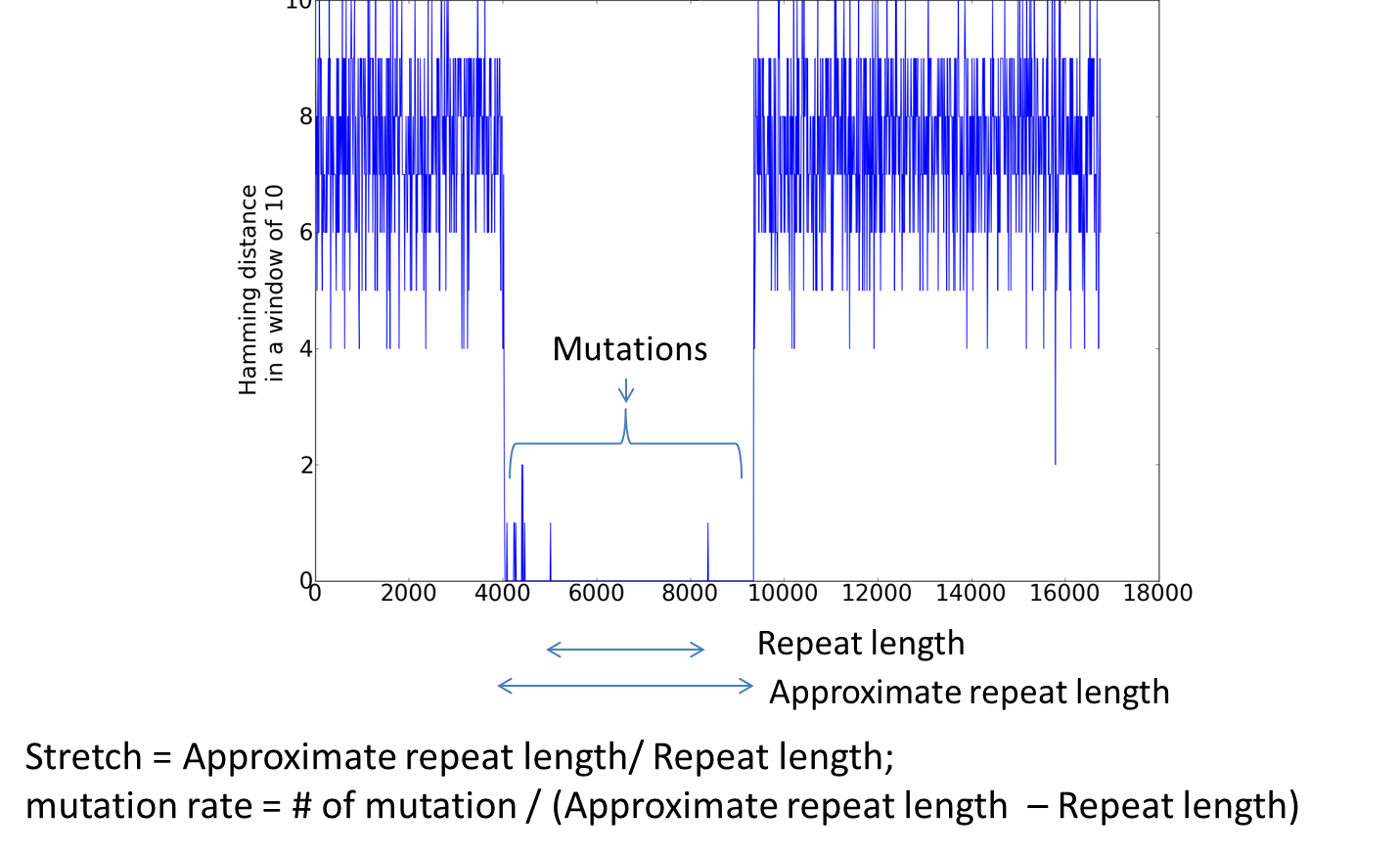}
\par\end{centering}

\caption{\label{fig:Example-of-how}Example of how to define stretch and mutation
rate}
\end{figure}

Moreover, we do a scatter plot to classify the approximate repeats(approximate
repeat having exact repeat length within top 20) and we have a plot
of approximate repeat spectrum as in Fig \ref{fig:Classification-of-approximate}. 

From the plots in Fig \ref{fig:Classification-of-approximate}, we
classify approximate repeat as homologous repeat if the stretch is
bigger than 1.25 and as non-homologous repeat if the stretch is less
than 1.25. 

For the scatter plot, every approximate repeat is a dot there with
x coordinate and y coordinate being mutation rate and stretch respectively.
And the color represents the length of that approximate repeat. For
the approximate repeat spectrum plot, the red bar represent non-homogeneous
repeat while the blue bar represent homologous repeat. The green dotted
line indicates the length of the longest repeat.

We focus on genomes when the non-homologous repeat dominates, namely
the longest interleave and the longest triple repeats are non-homologous
because the stretch is relatively short which can be captured by our
generative model. We do not distinguish between the length of approximate
or exact repeat are considered to be the same and we do not distinguish
between the two in the discussion because of the small stretch.

\begin{figure}
\begin{centering}
\includegraphics[width=10cm]{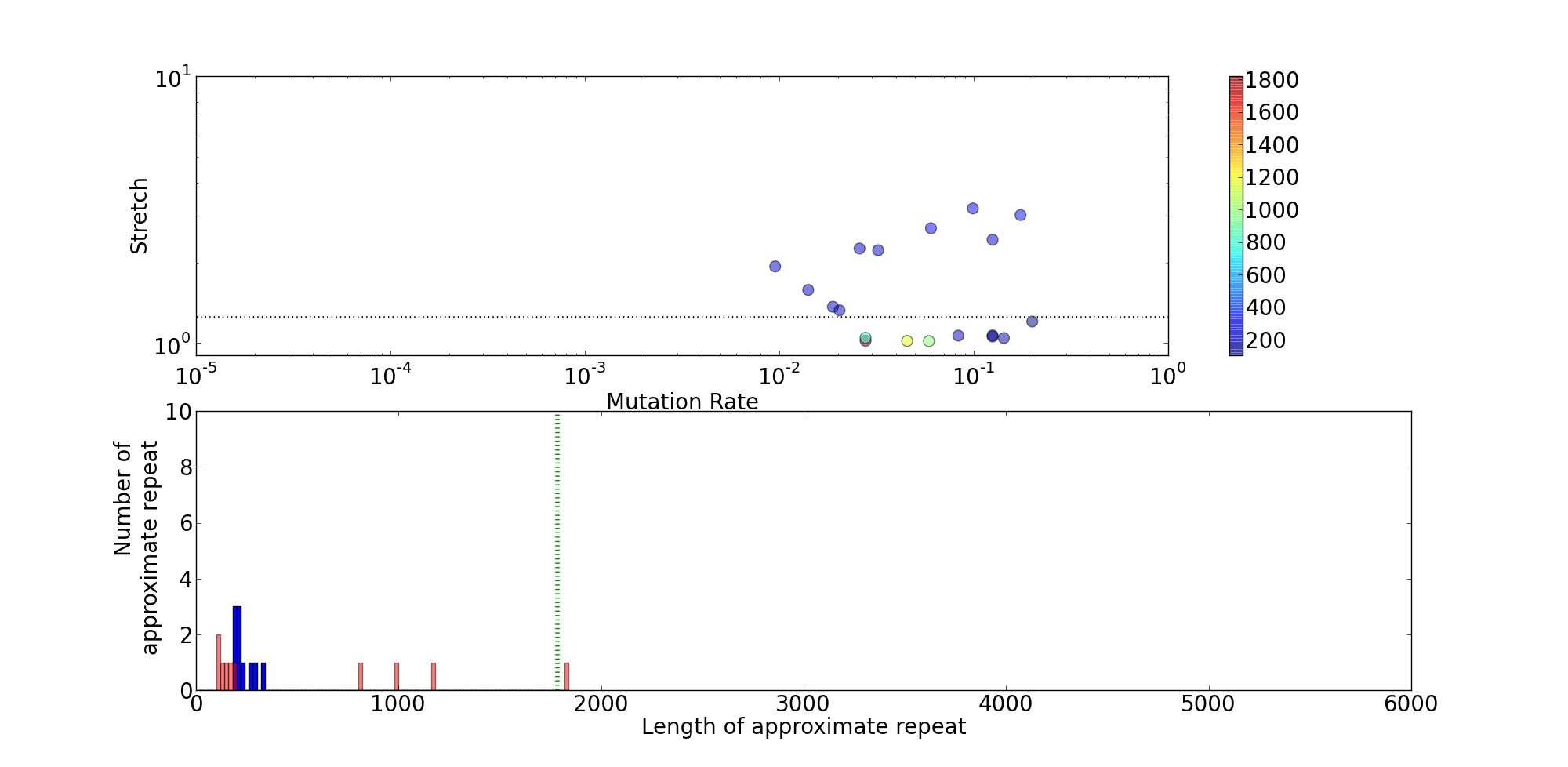}
\par\end{centering}

\caption{\label{fig:Classification-of-approximate}Classification of approximate
repeats and approximate repeat spectrum. The upper plot is scatter
plot to classify approximate repeat. The lower plot is the approximate
repeat spectrum}
\end{figure}

\subsection*{Stopping criterion for defining approximate repeat by MLE estimate}

\subsubsection*{Parametric model }

Let $L_{k}$ be the number of bases between the $(k-1)^{th}$ and
the $k^{th}$ SNPs starting from the right end-point of a repeat. 

We consider the following probabilistic model for the $L_{k}$. $\{L_{k}\}_{k=1}^{n}$is
taken as an independent sequence of geometrically distributed random
variables with parameter $\Theta=\{p_{1},p_{2},r\}$ defined as follows. 

\begin{eqnarray*}
L_{k} & \sim & \begin{cases}
Geo(p_{1}) & \text{if }1\le k\le r\\
Geo(p_{2}) & \text{if }r<k\le n
\end{cases}
\end{eqnarray*}

\subsubsection*{MLE estimate of parameters }

We now would like to estimate $\Theta$ given the observation of $\{\hat{L}_{k}\}_{k=1}^{n}$
by maximum likelihood estimation. Consider the log-likelihood function
$L(\Theta)=\log{\cal P}(\{\hat{L_{k}}\}_{k=1}^{n}\mid\Theta)$.
\begin{eqnarray}
L(\Theta) & = & \log{\cal P}(\{\hat{L_{k}}\}_{k=1}^{n}\mid\Theta)\\
 & = & \log\{[\Pi_{k=1}^{r}(1-p_{1})^{\hat{L_{k}}}p_{1}]\cdot[\Pi_{k=r+1}^{n}(1-p_{2})^{\hat{L_{k}}}p_{2}]\}\\
 & = & r\log p_{1}+[\sum_{k=1}^{r}\hat{L}_{k}]\cdot\log(1-p_{1})+(n-r)\cdot\log p_{2}+[\sum_{k=r+1}^{n}\hat{L}_{k}]\cdot\log(1-p_{2})
\end{eqnarray}

And we want to find $\hat{\Theta}=\arg\max_{\Theta}L(\Theta)$.

Observe that, if we fix $1\le r\le n$, then the optimal $\hat{p}_{1}$
and $\hat{p}_{2}$ can be readily obtained by taking derivative on
$L(\Theta)$ with respect to $p_{1}$and $p_{2}$, specifically,
\begin{eqnarray}
\hat{p_{1}} & = & \frac{1}{1+\frac{\sum_{k=1}^{r}\hat{L}_{k}}{r}}\\
\hat{p}_{2} & = & \frac{1}{1+\frac{\sum_{k=r+1}^{n}\hat{L}_{k}}{n-r}}
\end{eqnarray}

$\hat{\Theta}$ can then be obtained by running over all integral
$1\le r\le n$ and use the corresponding optimal $\hat{p}_{1}$ and
$\hat{p}_{2}$ to obtain the $L(\Theta)$, and finally we use the
$r$ that gives the highest value of $L(\Theta)$ as the MLE estimate
given the observation.

\subsubsection*{Linear time algorithm to estimate the stopping criterion}

Moreover, this can be done by the following algorithm Algo \ref{alg:Linear-time-algorithm},
which run in linear time $\Theta(n)$ with respect to the number of
observations $n$.

\begin{algorithm}

\raggedright{}
1.

a)$C_{0}\leftarrow0$

b) for r = 1 to n

$C_{r}\leftarrow C_{r}+\hat{L}_{r}$

2.

a)$D_{0}\leftarrow C_{n}$

b)for r = 1 to n

$D_{r}\leftarrow C_{n}-\hat{L}_{r}$

3. for r = 1 to n

$\hat{p}_{1}^{(r)}\leftarrow\frac{1}{1+\frac{C_{r}}{r}}$

$\hat{p}_{2}^{(r)}\leftarrow\frac{1}{1+\frac{D_{r}}{n-r}}$

$\Theta_{r}\leftarrow(r,\hat{p}_{1}^{(r)},\hat{p}_{2}^{(r)})$

$X_{r}\leftarrow L(\Theta_{r})$

4. find maximum among $\{X_{r}\}_{r=1}^{n}$, and the corresponding
$\Theta_{r}$ is the MLE estimate. 

5. (Differentiate between homologous and non-homologous repeat)

If the optimal $\hat{p}_{1}^{(r)},\hat{p}_{2}^{(r)}$are too close
(i.e. $\hat{p}_{1}^{(r)}>0.2$), then claim $r=1$; else, claim $\hat{r}=r$. 

\caption{\label{alg:Linear-time-algorithm}Linear time algorithm to estimate
the stopping criterion}
\end{algorithm}

A sample plot is of who we can use the critierion to accurately define
the ending of approximate repeat is shown in Fig \ref{fig:An-example-plot}. 

\begin{figure}
\includegraphics[width=12cm]{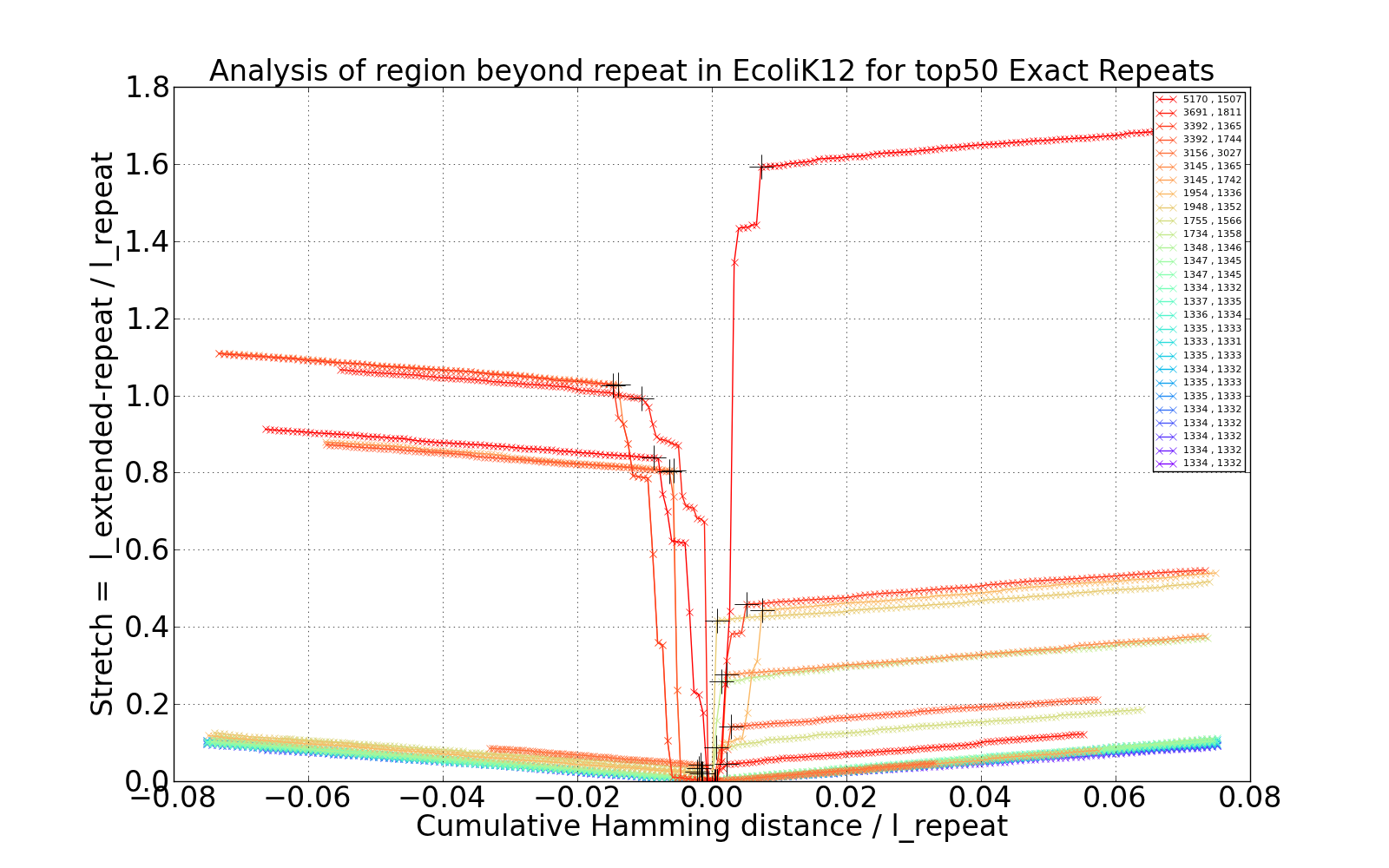}

\caption{\label{fig:An-example-plot}An example plot that define the stopping
point of approximate repeat by Algorithm \ref{alg:Linear-time-algorithm}}
\end{figure}

\newpage{}

\section*{Appendix of the dot plot of finished genomes}
\begin{figure}
\subfloat[Index 1]{\includegraphics[width=4cm]{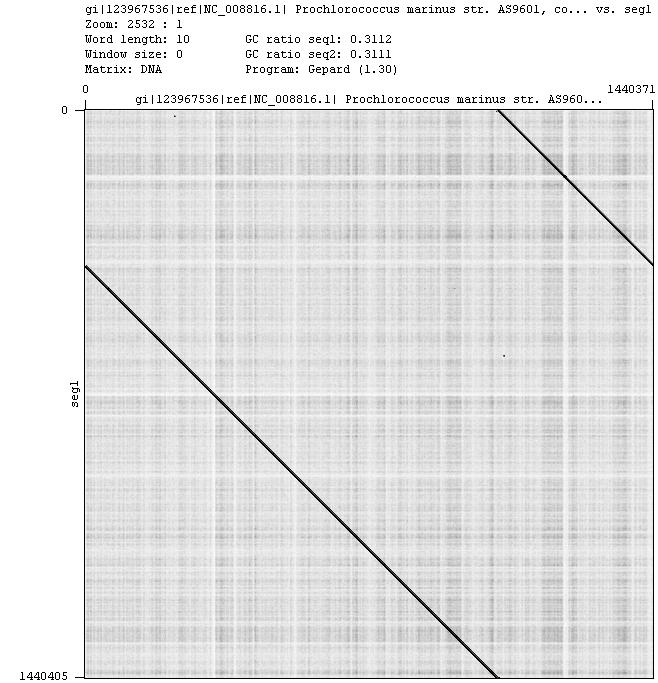}}
\subfloat[Index 2]{\includegraphics[width=4cm]{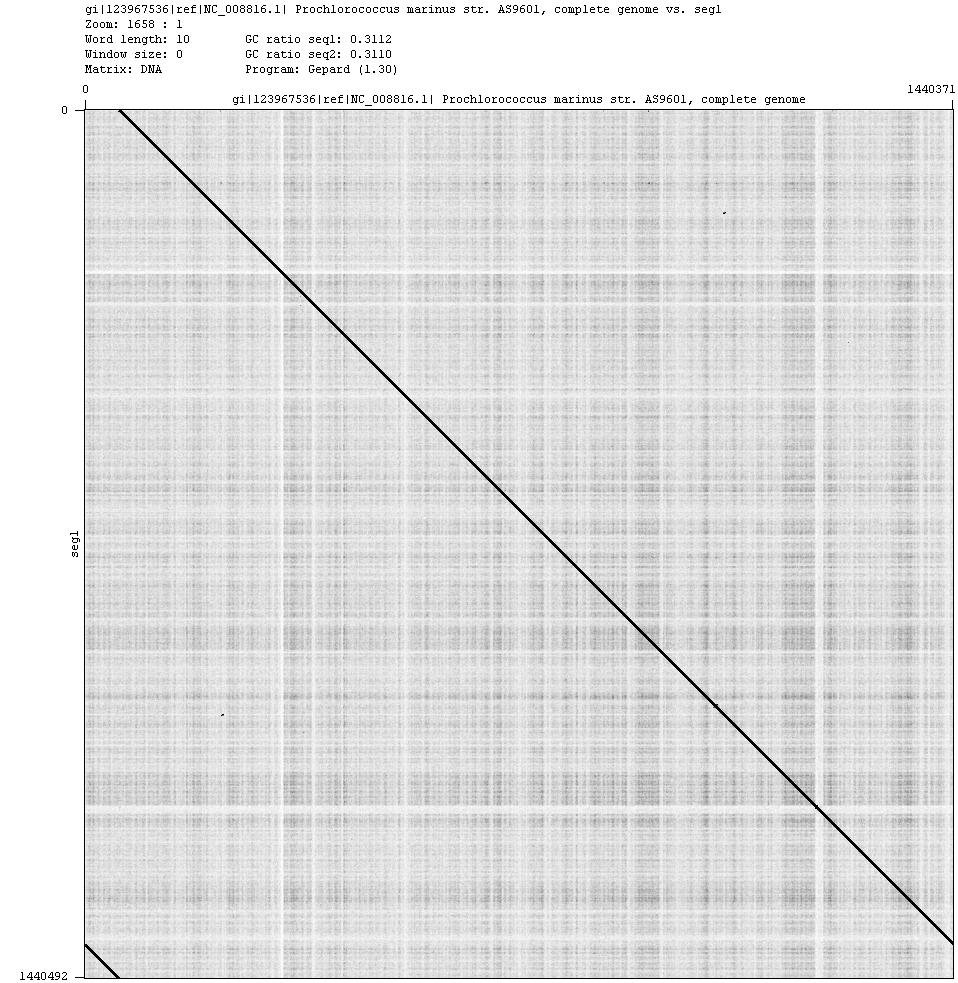}}
\subfloat[Index 3]{\includegraphics[width=4cm]{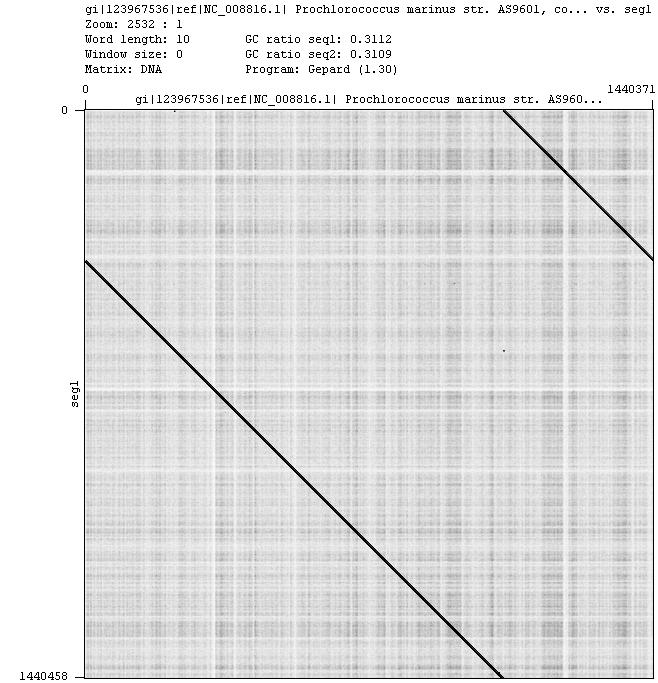}}

\subfloat[Index 4]{\includegraphics[width=4cm]{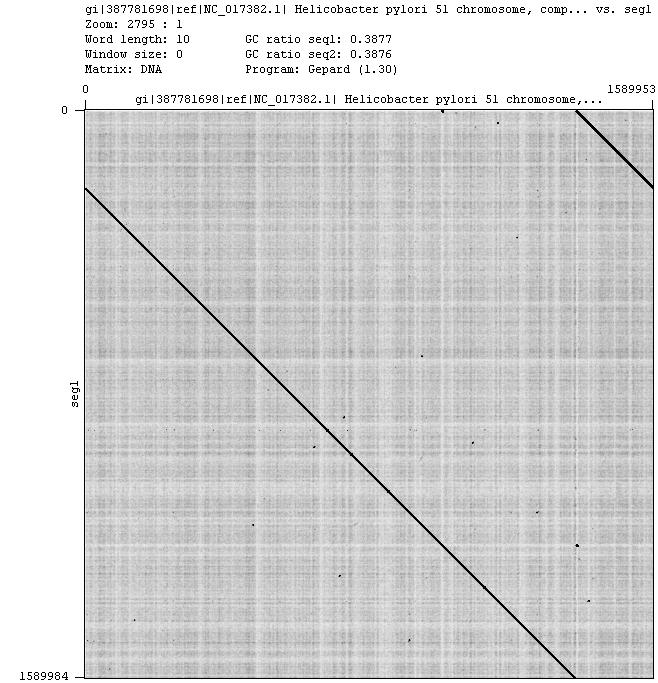}}
\subfloat[Index 5]{\includegraphics[width=4cm]{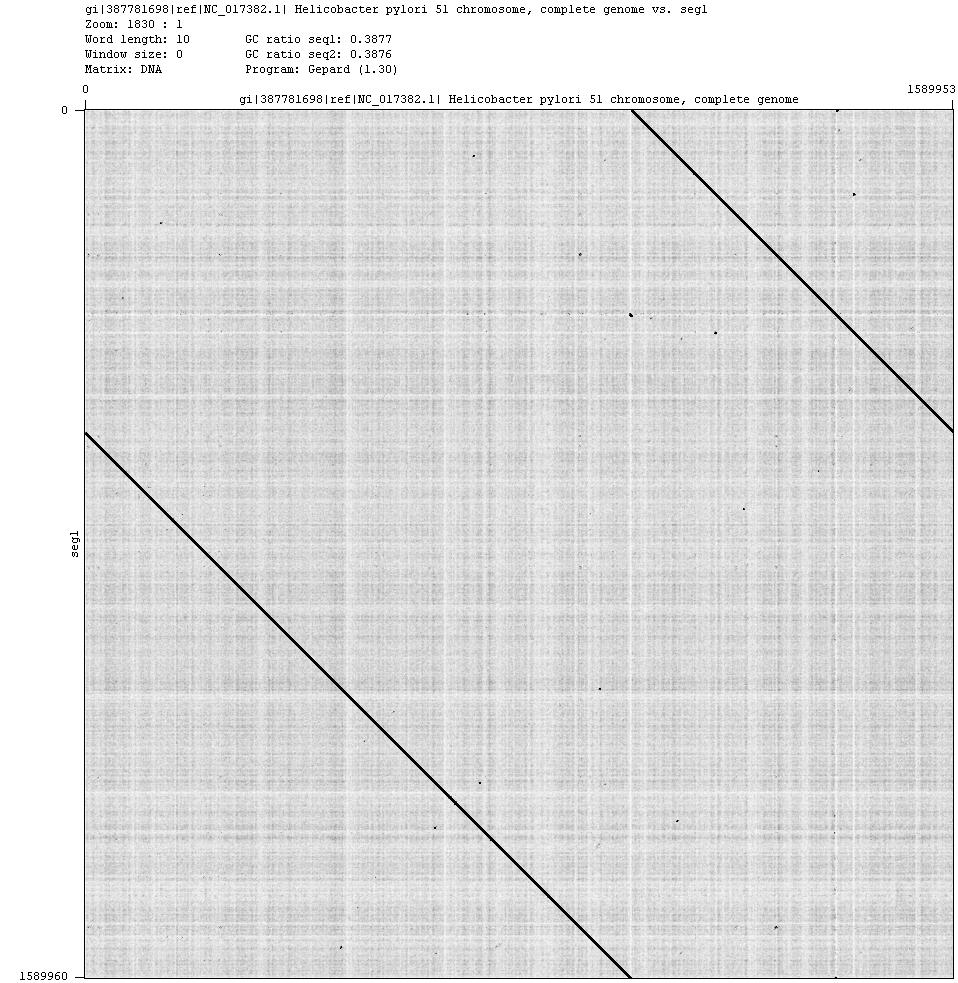}}
\subfloat[Index 6]{\includegraphics[width=4cm]{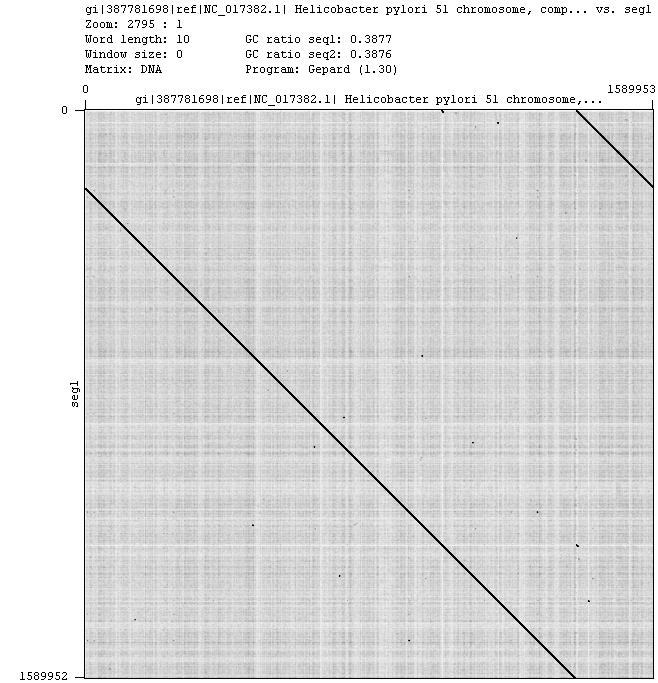}}

\subfloat[Index 7]{\includegraphics[width=4cm]{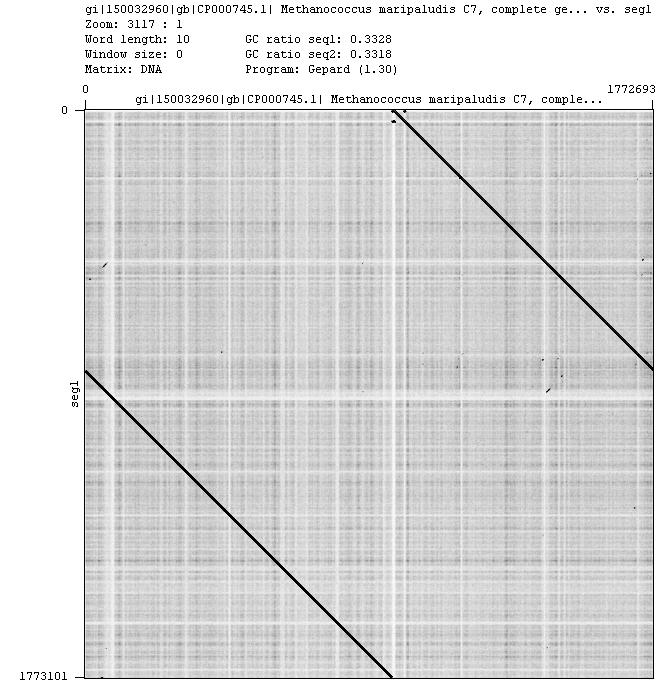}}
\subfloat[Index 8]{\includegraphics[width=4cm]{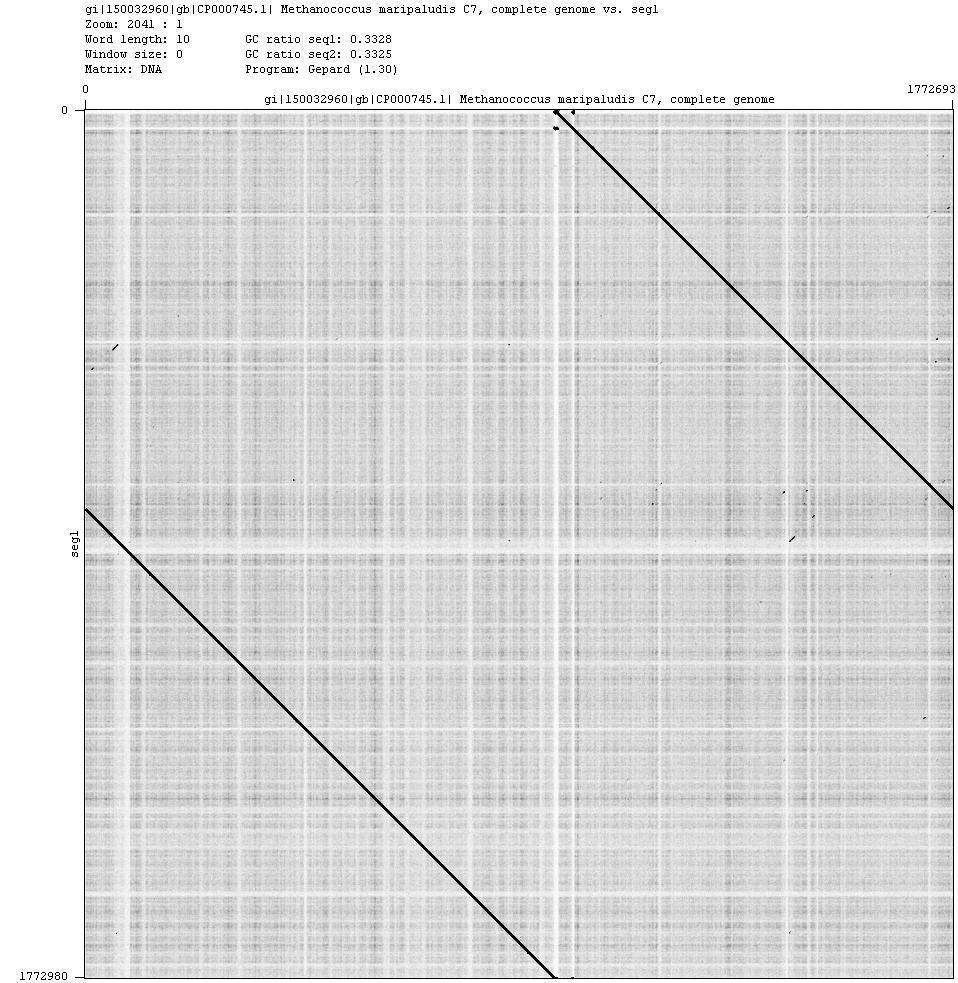}}
\subfloat[Index 9]{\includegraphics[width=4cm]{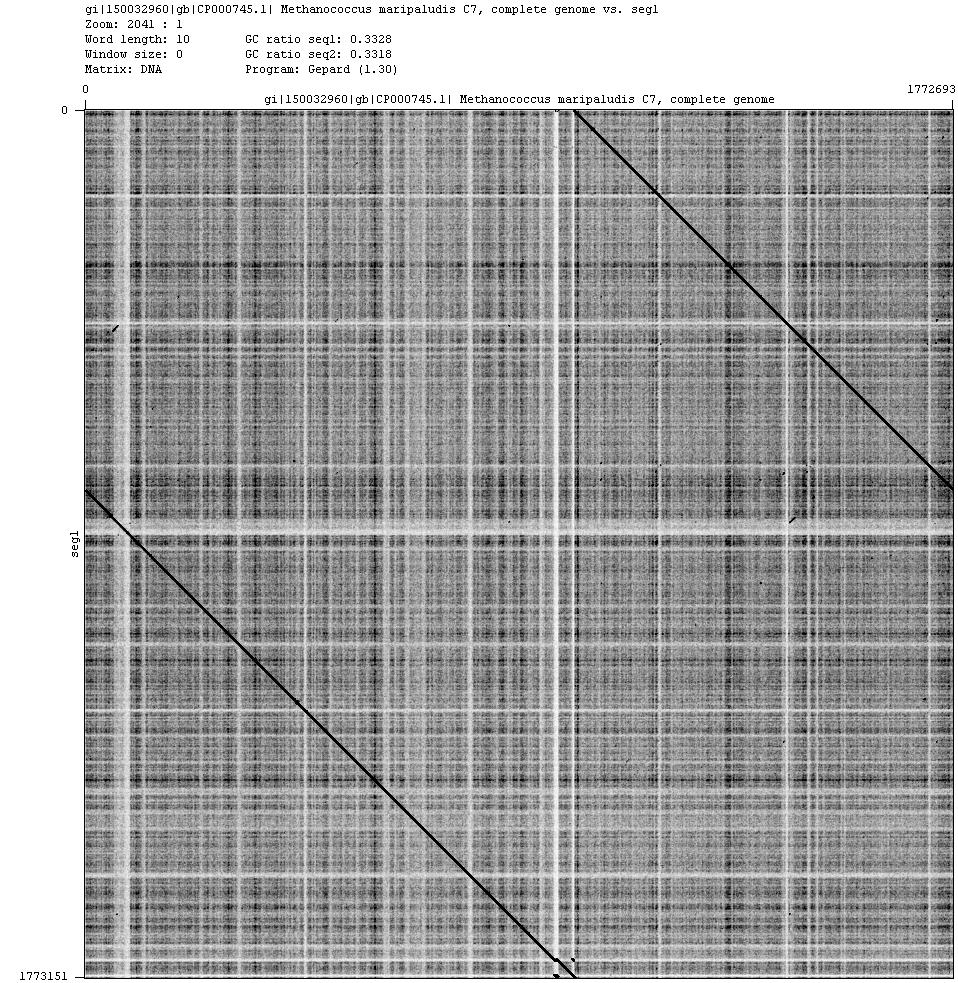}}

\subfloat[Index 10]{\includegraphics[width=4cm]{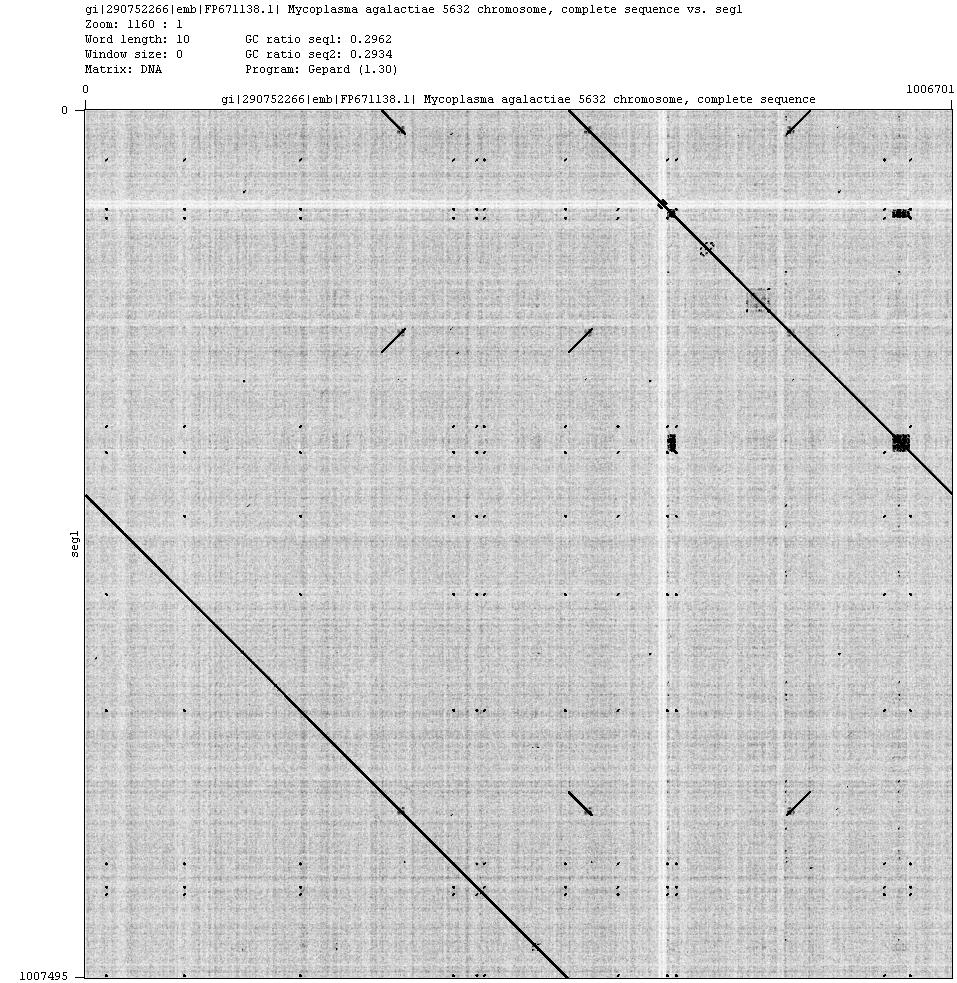}}
\subfloat[Index 11]{\includegraphics[width=4cm]{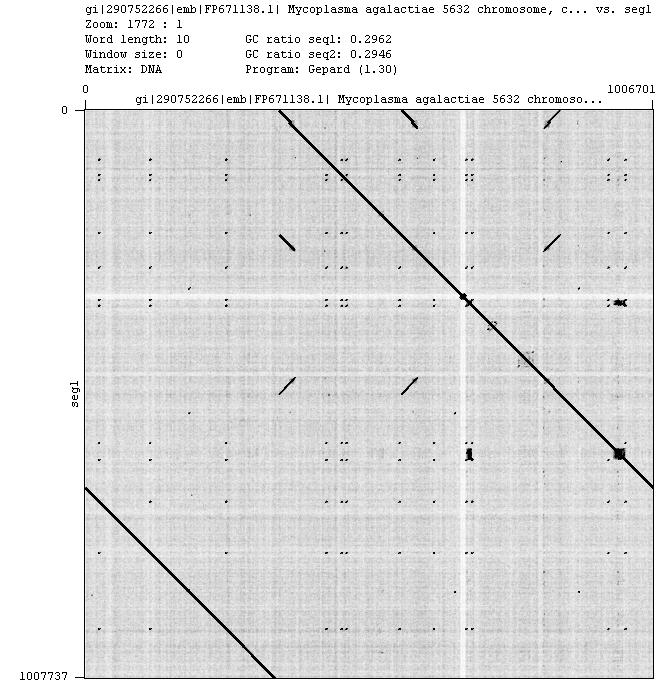}}
\subfloat[Index 12]{\includegraphics[width=4cm]{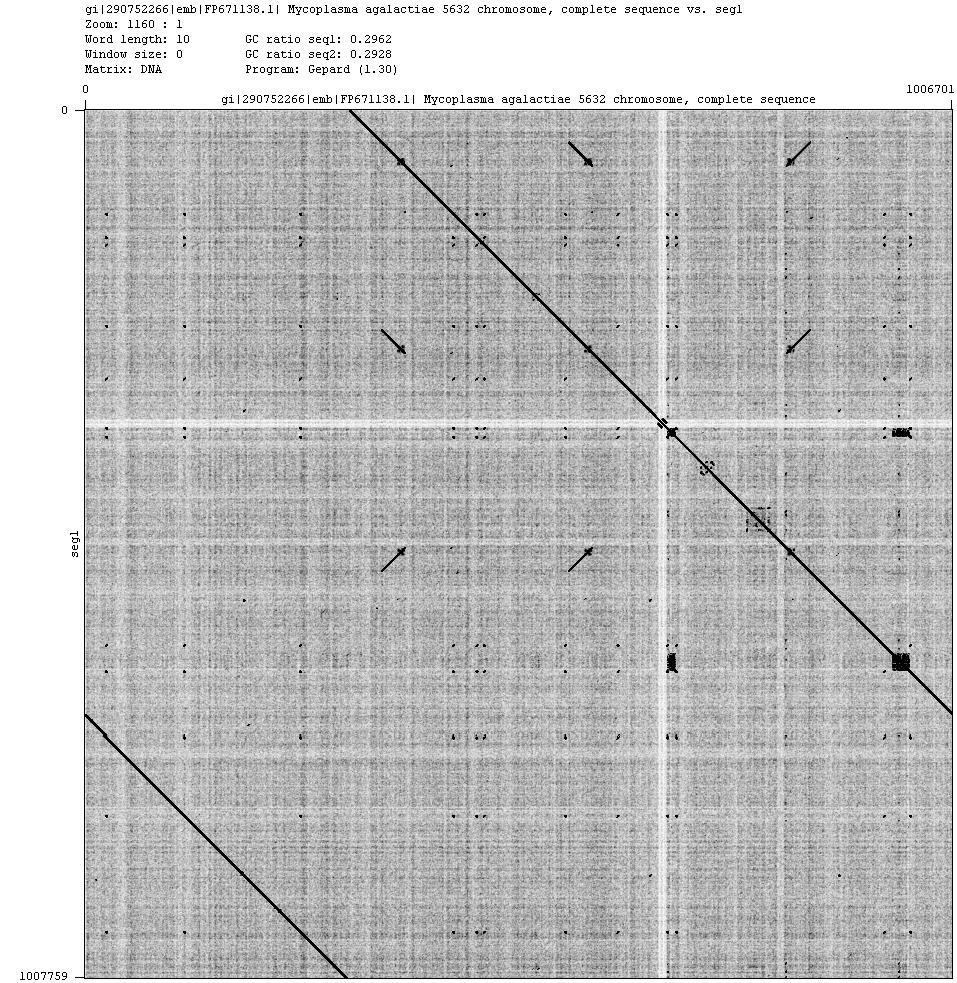}}

\caption{\label{fig:Dot-plot-of}Dot plot of recovered genomes against ground
truth(according to index in Table \ref{tab:Assembly-of-several}) }
\end{figure}

\end{document}